\documentclass[11pt]{article} 
\usepackage{times}
\usepackage{amsmath,amsfonts}
\usepackage{epsfig,amssymb,amstext,xspace}
\usepackage{theorem} 
\usepackage{calc}

\newcommand{\p}{{\em P}\xspace}
\newcommand{\np}{{\em NP}\xspace}
\newcommand{\nphard}{\np-hard\xspace} 

\newcommand{\apx}{{\em APX}\xspace}

\DeclareMathOperator*{\Exp}{E}
\newcommand{\E}[2][{}]{\ensuremath{{\textstyle \Exp_{#1}}\left[#2\right]}}

\newtheorem{theorem}{Theorem}[section]
\newtheorem{lemma}[theorem]{Lemma}

\newtheorem{claim}[theorem]{Claim}

\newtheorem{corollary}[theorem]{Corollary}

{\theorembodyfont{\rmfamily} \newtheorem{remark}[theorem]{Remark}}

\def\blksquare{\rule{2mm}{2mm}}
\def\qedsymbol{\blksquare}
\newcommand{\bg}[1]{\medskip\noindent{\bf #1}}
\newcommand{\ed}{{\hfill\qedsymbol}\medskip}
\newenvironment{proof}{\bg{Proof : }}{\ed}
\newenvironment{proofnobox}{\bg{Proof : }}{\medskip}

\newcommand{\R}{\ensuremath{\mathbb R}}

\newcommand{\T}{\ensuremath{\mathcal T}}

\newcommand{\Pc}{\ensuremath{\mathcal P}}

\newcommand{\Sc}{\ensuremath{\mathcal S}}
\newcommand{\V}{\ensuremath{\mathcal V}}
\newcommand{\OPT}{\ensuremath{\mathit{OPT}}}
\newcommand{\sm}{\ensuremath{\setminus}}
\newcommand{\es}{\ensuremath{\emptyset}}

\newcommand{\size}[1]{\ensuremath{\left|#1\right|}}
\newcommand{\ceil}[1]{\ensuremath{\left\lceil#1\right\rceil}}

\newcommand{\poly}{\operatorname{poly}}
\newcommand{\polylog}{\operatorname{polylog}}

\newcommand{\junk}[1]{}

\newcommand{\sse}{\subseteq}
\newcommand{\union}{\cup}

\newcommand{\assign}{\ensuremath{\leftarrow}}

\newcommand{\al}{\ensuremath{\alpha}}
\newcommand{\sg}{\ensuremath{\sigma}}
\newcommand{\ld}{\ensuremath{\lambda}}
\newcommand{\Om}{\ensuremath{\Omega}}
\newcommand{\dt}{\ensuremath{\delta}}
\newcommand{\Dt}{\ensuremath{\Delta}}
\newcommand{\e}{\ensuremath{\epsilon}}
\newcommand{\ve}{\ensuremath{\varepsilon}}
\newcommand{\gm}{\ensuremath{\gamma}}
\newcommand{\Gm}{\ensuremath{\Gamma}}

\newcommand{\ov}[1]{\ensuremath{\overline{#1}}}
\newcommand{\del}{\ensuremath{\partial}}
\newcommand{\lb}{\ensuremath{\mathit{lb}}}
\newcommand{\ub}{\ensuremath{\mathit{ub}}}
\newcommand{\vol}{\ensuremath{\mathrm{vol}}}

\newcommand{\mc}{\ensuremath{\mathsf{MC}}\xspace}
\newcommand{\kmc}{\ensuremath{k}-\mc}
\newcommand{\nmc}{\ensuremath{\mathsf{NMC}}\xspace}
\newcommand{\knmc}{\ensuremath{k}-\nmc}
\newcommand{\edknmc}{\ensuremath{\mathsf{ED}}-\knmc}
\newcommand{\ndknmc}{\ensuremath{\mathsf{ND}}-\knmc}
\newcommand{\mwc}{\ensuremath{\mathsf{MWC}}\xspace}
\newcommand{\kmwc}{\ensuremath{k}-\mwc}
\newcommand{\mic}{\ensuremath{(s,t)\text{-}\mathsf{Cut}}\xspace}
\newcommand{\mfip}{\ensuremath{\mathsf{MFIP}}\xspace}
\newcommand{\kmic}{\ensuremath{k}-\mic}
\newcommand{\mcalg}{\ensuremath{\mathsf{MCAlg}}\xspace}
\newcommand{\nmcalg}{\ensuremath{\mathsf{NMCAlg}}\xspace}
\newcommand{\twomcalg}{\ensuremath{2\text{-}}\mcalg}
\newcommand{\kmcalg}{\ensuremath{k\text{-}}\mcalg}
\newcommand{\edtwonmcalg}{\ensuremath{\mathsf{ED}\text{-}2\text{-}}\nmcalg}
\newcommand{\ndknmcalg}{\ensuremath{\mathsf{ND}\text{-}k\text{-}}\nmcalg}
\newcommand{\ssve}{\ensuremath{\mathsf{SSVE}}\xspace}
\newcommand{\dks}{\ensuremath{\mathsf{D}k\mathsf{S}}\xspace}
\newcommand{\mindks}{\ensuremath{\mathsf{Min}}\dks}
\newcommand{\minrep}{\ensuremath{\mathsf{MinRep}}\xspace}

\newcommand{\optset}{\ensuremath{O^*}}
\newcommand{\bG}{\ensuremath{\ov{G}}}
\newcommand{\bE}{\ensuremath{\ov{E}}}
\newcommand{\tG}{\ensuremath{\tilde G}}
\newcommand{\cl}{\ensuremath{\sg}}

\newcommand{\ignore}[1]{}

\title{Improved Region-Growing and Combinatorial Algorithms for $k$-Route Cut Problems}

\author{
    Guru Guruganesh\thanks{{\tt  ggurugan@cs.cmu.edu}.
    Department of Computer Science, Carnegie Mellon University, Pittsburgh, USA.
    Work done while the author was an undergraduate research assistant at the University of Waterloo under Prof. Sanita.}
\and
    Laura Sanita\thanks{{\tt sanita@uwaterloo.ca}.
    Dept. of Combinatorics and Optimization, Univ. Waterloo, Waterloo, ON N2L 3G1.
    }
\and
    Chaitanya Swamy\thanks{{\tt cswamy@math.uwaterloo.ca}.  
    Dept. of Combinatorics and Optimization, Univ. Waterloo, Waterloo, ON N2L
    3G1. Supported in part by NSERC grant 327620-09, an NSERC Discovery Accelerator
    Supplement Award and an Ontario Early Researcher Award.}  
}

\date{} 

\begin{document}

\maketitle

\begin{abstract}
We study the {\em $k$-route} generalizations of various cut problems, 
the most general of which is \emph{$k$-route multicut} (\kmc) problem, wherein we have $r$
source-sink pairs and the goal is to delete a minimum-cost set of edges to 
reduce the edge-connectivity of every source-sink pair to below $k$.
The $k$-route extensions of multiway cut (\kmwc), and the minimum $s$-$t$ cut problem
(\kmic), are similarly defined.  
We present various approximation and hardness results for  \kmc, \kmwc, and \kmic that
improve the state-of-the-art for these problems in several cases.
Our contributions are threefold.
\begin{list}{$\bullet$}{\itemsep=0ex \addtolength{\leftmargin}{-2ex}}
\item For {\em $k$-route multiway cut}, we devise simple, but surprisingly effective, 
combinatorial algorithms that yield bicriteria approximation guarantees 
that markedly improve upon the previous-best guarantees. 

\item For {\em $k$-route multicut}, we design algorithms that improve upon
the previous-best approximation factors by roughly an $O(\sqrt{\log r})$-factor, when
$k=2$, and for general $k$ and unit costs and any fixed violation of the connectivity
threshold $k$. 
The main technical innovation is the definition of a new, powerful
\emph{region growing} lemma that allows us to perform region-growing in a recursive
fashion even though 
the LP solution yields a {\em different metric} for each source-sink pair, and 
{\em without incurring an $O(\log^2 r)$ blow-up} in the cost that is inherent in some
previous applications of region growing to $k$-route cuts. 
We obtain the same benefits as~\cite{EvenNRS00} do
in their divide-and-conquer algorithms, and thereby obtain an 
$O(\ln r\ln\ln r)$-approximation to the cost.
We also obtain some extensions to $k$-route node-multicut problems.

\item We complement these results by showing that the {\em $k$-route $s$-$t$ cut}
problem is at least as hard to approximate as the {\em densest-$k$-subgraph} (\dks)
problem on uniform hypergraphs. In particular, this implies that one cannot avoid a
$\poly(k)$-factor if one seeks a unicriterion approximation, without improving the
state-of-the-art for \dks on graphs, and proving the existence of a family of one-way
functions. Previously, only \nphard{}ness of \kmic was known.
\end{list}
\end{abstract}

\section{Introduction}
The problem of finding minimum size cuts for a given graph has a rich history in the field of combinatorial optimization,
with a wide range of applications in logistics, transportation and telecommunication systems.
One key problem of interest is that of disconnecting a given set of node pairs in a
network by removing edges at minimum cost. 
Formally, in the \emph{multicut} problem, we are given
an undirected graph $G=(V,E)$ with nonnegative edge costs $\{c_e\}_{e\in E}$ and
pairs of nodes $(s_1,t_1),\ldots,(s_r,t_r)$ called source-sink pairs or commodities, and
we seek a minimum-cost set of edges whose removal disconnects every $s_i$-$t_i$ pair.
Two special cases of this problem have by themselves attracted widespread attention:
(i) the celebrated \emph{minimum $s$-$t$ cut} problem, which is the special case when $r=1$; 
and (ii) the \emph{multiway cut} problem~\cite{DahlhausJPSY92}, 
where every pair of nodes from a given set $T\sse V$ of terminals forms a commodity.
These cut problems and their variants have been widely studied in terms of hardness and
approximation (see, e.g.,~\cite{Vazirani01,WilliamsonS10}), 
have numerous direct applications (e.g., identifying bottlenecks in a network), 
and algorithms for 
them serve as important 
primitives in the design of divide-and-conquer algorithms (see, e.g.,~\cite{EvenNRS00,Shmoys97}) 
and find application in diverse settings such as image segmentation, VLSI design and network routing 
(see, e.g.,~\cite{LeightonR99, MalikS00, BhattL84, Racke02}). 

We study a natural generalization of the above cut problems 
motivated by the fact that in various settings,
we are not interested in a complete disconnection of our terminals but rather in reducing
their connectivity below a certain threshold. 
Specifically, in the {\em $k$-route multicut} (\kmc) problem, the input is a multicut
instance and an integer $k\geq 1$; 
the goal is to find a minimum-cost set $F\sse E$ of edges so that there are at most
$(k-1)$ edge-disjoint $s_i$-$t_i$ paths in $(V,E\sm F)$ for all $i=1,\ldots,r$. 
We define the {\em $k$-route multiway cut} (\kmwc), and the {\em $k$-route $(s,t)$-cut}
(\kmic) problems analogously.

The study of $k$-route cut problems can be motivated from various perspectives.
One motivation comes from the fact that $k$-route cuts are dual objects to 
{\em $k$-route flows}~\cite{Kishimoto96}, 
which can be seen as a robust or fault-tolerant version of flows 
where we seek to send traffic along tuples of $k$ edge-disjoint paths. 
A $k$-route cut establishes an upper bound on the value (suitably defined) of the
maximum $k$-route flow, and can thus 
be seen as identifying the bottleneck in a network when we seek a certain level of
robustness. 
$k$-route cut problems can also be directly motivated as 
abstracting the problem of an
attacker who seeks to reduce connectivity in a given network 
while incurring minimum cost. Viewed from this perspective, $k$-route cut problems
are closely related to {\em network interdiction} problems, which typically consider the
complementary objective of minimizing source-sink connectivity subject to a budget
constraint on the edge-removal cost~\cite{Phillips93,Wood93,Zenklusen10}.  

The $k$-route cut problems are at least as hard as their 1-route counterparts.
Multicut and multiway cut are \apx-hard~\cite{DahlhausJPSY92}, with the former 
not admitting any constant-factor approximation assuming the unique-games
conjecture~\cite{ChawlaGR08}, and \kmic is \nphard; 
hence, we focus on approximation algorithms. 
Moreover, as highlighted in~\cite{ChekuriK08,BarmanC10,KolmanS11,ChuzhoyMVZ12,KolmanS12},  
$k$-route cut problems turn out to be much more challenging than their 1-route
counterparts, especially for non-constant $k$, so (as in~\cite{ChuzhoyMVZ12}) we consider  
bicriteria approximation guarantees. (This is further justified by our hardness result for
\kmic in Section~\ref{kstcut}.)  
We say that a solution $F\sse E$ is an 
$(\alpha,\beta)$-approximation for the given \kmc instance if $\sum_{e\in F}c_e$ is at most
$\beta$ times the optimal value, and $(V,E\setminus F)$ contains at most 
$\alpha(k -1)$ edge-disjoint $s_i$-$t_i$ paths for all $i=1,\ldots,r$.

\vspace{-1ex}
\paragraph{Our results.} 
We develop various approximation and hardness results for  \kmc, \kmwc\ and \kmic that
improve upon 
the current-best approximation and hardness results in several cases. 

In Section~\ref{kmwc}, we consider the $k$-route multiway cut problem. We devise an
$\bigl(O(1),O(1)\bigr)$-approximation for \kmwc with unit costs (Theorem~\ref{unitthm}),
and an $\bigl(O(1),O(\log r)\bigr)$-approximation with general costs
(Theorem~\ref{genthm}), where $r=|T|$. 
The previous-best guarantees for \kmwc (for general $k$) are those that
follow from the results of Chuzhoy et al.~\cite{ChuzhoyMVZ12} for \kmc, namely, 
an $\bigl(O(1),O(\log^{1.5}r)\bigr)$-approximation for unit costs and an 
$\bigl(O(\log r),O(\log^3 r)\bigr)$-approximation for general costs. 
Thus, our guarantees 
constitute a significant improvement in the state-of-the-art for \kmwc.
We also show that the special case where $T=V$, which we call 
{\em $k$-route all-pairs cut}, is \apx-hard for $k\geq 3$
(Appendix~\ref{append-allpairs}). (For $k=1, 2$, it is easy to see that all-pairs
$k$-route cut is polytime solvable.)

In Section~\ref{kmc}, we 
design algorithms for the $k$-route multicut problem.
We achieve approximation ratios of 
$O(\ln r\ln\ln r)$ for 2-\mc, and 
$\bigl(\gm,O(\frac{\gm}{(\sqrt{\gm}-1)^2}\ln r\ln\ln r)\bigr)$ 
for \kmc with unit costs. 
In contrast, Chuzhoy et al.~\cite{ChuzhoyMVZ12} obtain approximation ratios of
$O(\log^{1.5}r)$ for 2-\mc, and $\bigl(\gm,O(\frac{\log^{1.5}r}{\min\{1,\gm-1\}})\bigr)$
for \kmc with unit costs.
Thus, for any fixed $\gm$ (i.e., independent of $k$ and $r$), our results
improve upon the previous-best guarantees for these cases in~\cite{ChuzhoyMVZ12} by
roughly an $O(\sqrt{\log r})$-factor. 
(Setting $\gm=\frac{k}{k-1}$, our guarantee and the one in~\cite{ChuzhoyMVZ12}
become unicriterion approximations that are incomparable.)
In contrast to the algorithms in~\cite{ChuzhoyMVZ12}, which rely
on approximations to suitable variants of sparsest cut, we devise
{\em rounding} algorithms for a natural LP-relaxation for \kmc, 
and our guarantees therefore also translate to integrality-gap results.
In Section~\ref{node-extn}, we consider some extensions to $k$-route {\em node-multicut}  
problems. 

Complementing the above results, we show in Section~\ref{kstcut} that 
\kmic is at least as hard as the 
densest $k$-subgraph (\dks) problem: a $\rho$-approximation for \kmic yields a
$(2\rho^\ld)$-approximation for \dks on $\ld$-uniform hypergraphs (Theorem~\ref{kmichard}).
The latter problem is hard to approximate within an $n^{\e_0}$-factor, for some constant
$\e_0$, for all $\ld\geq 3$, unless a certain family of one-way functions
exists~\cite{Applebaum11}. This implies that obtaining a unicriterion
$O\bigl(k^{\e_0}\polylog(n)\bigr)$-approximation (even) for \kmic for some constant $\e_0$
would improve the state-of-the-art for the notoriously hard densest $k$-subgraph problem
on graphs, 
and imply the existence of certain one-way functions.  
Previously, 
only \nphard{}ness of \kmic was known, as a consequence of the fact that certain
\nphard unbalanced graph partitioning problems~\cite{HayrapetyanKPS05, LiZ13} can be cast
as special cases of \kmic. 

\vspace{-1ex}
\paragraph{Our techniques.}
Our algorithms 
for \kmwc are combinatorial, and 
rely on the following simple, but quite useful observation: if $F\sse E$ is
feasible, then $\bG=(V,E\sm F)$ has a multiway cut with at most $(k-1)(r-1)$ edges
(Claim~\ref{kmwc-prop}). 
Using this, we show that we can identify a terminal $t_i\in T$, a $t_i$-isolating cut, and a
set of edges of cost $O\bigl(\frac{\OPT}{|T|}\bigr)$ whose removal causes the
$t_i$-isolating cut to have $O(k-1)$ edges. 
We include these edges, drop $t_i$ from $T$, and repeat, which naturally
yields an $O(\log r)$-approximation in the cost. The improvement for unit costs stems 
from the stronger property that either the minimum multiway-cut in $G$ has cost $O(\OPT)$,
or there is some $t_i$-isolating cut of value $O(k-1)$; thus, we may now drop terminals incurring
{\em zero cost}, which results in an improved $O(1)$ cost-approximation.

Interestingly, \cite{BarmanC10} 
use a similar approach to obtain an $\bigl(O(1),O(1)\bigr)$-approximation for
single-source \kmc with unit costs 
and they remark that such an approach is unlikely to work for \kmwc because there are
examples where every pair of terminals is $2(k-1)$-edge connected but the optimal multiway
cut value is $\Omega(r)\cdot\OPT$.
Thus, a useful insight to emerge from our work is that whereas a $2$-factor violation in 
the pairwise terminal 
connectivity does not ensure that the multiway cut value is
$O(\OPT)$, a $(2+\ve)$-factor violation in connectivity does, for any $\ve>0$.

\smallskip
Our algorithms for \kmc are based on rounding an optimal solution to a natural
LP-relaxation of the problem. This is technically the most sophisticated part of the
paper. 
The main technique that we use is {\em region growing}. 
The idea is to view the LP
solution as a metric, grow a suitable ball in this metric and prove 
a region-growing lemma showing 
that the cost of the ball-boundary edges can be charged to the ball-volume, where
volume measures the contribution to the LP objective from the edges 
inside the ball. 
This was introduced by~\cite{LeightonR99,GargVY96} in the context of the sparsest cut and
multicut problems, and Even et al.~\cite{EvenNRS00}, building upon the work of
Seymour~\cite{Seymour95}, extended the technique to obtain improved guarantees for
various divide-and-conquer algorithms that involve recursive applications of 
region growing.  
However, in contrast with various applications of region growing considered
in~\cite{LeightonR99,GargVY96,EvenNSS98,EvenNRS00}, the difficulty in the $k$-route multicut problem
stems from the fact that an LP-solution yields a {\em different metric for each source-sink
pair} instead of a single common metric that can be applied in the region-growing
process. (In particular, \kmc does not fall into the divide-and-conquer framework of Even
et al.~\cite{EvenNRS00}.) 
Although~\cite{BarmanC10,KolmanS11,KolmanS12} adapted the region-growing lemma
in~\cite{LeightonR99,GargVY96} to the 2-route, 3-route, and the $k$-route single-source
settings, their approach seems incapable of obtaining any thing better than an 
$O(\log^2 r)$-approximation---one loses one $\log$-factor due to region growing and
another due to recursion---which is worse than the guarantees 
in~\cite{ChuzhoyMVZ12}. 
(In fact, \cite{ChuzhoyMVZ12} abandoned the region-growing approach and 
used a greedy set-cover strategy to obtain their improvements
over~\cite{BarmanC10,KolmanS11,KolmanS12}.)  

Our chief technical novelty is to prove a region-growing lemma (see Lemmas~\ref{reggrow}
and~\ref{cor2}) applicable to settings with different metrics, that is inspired by, but 
more general, than the analogous lemma in~\cite{EvenNRS00}, and much more
sophisticated than the one used in~\cite{BarmanC10,KolmanS11,KolmanS12}. 
This lemma, coupled with a subtle insight about the metrics derived from the LP solution, 
allows us to obtain the same kind of savings in our recursive region-growing algorithm
that Even et al.~\cite{EvenNRS00} obtain (via their region-growing lemma) in their
divide-and-conquer algorithms; this yields our improved approximation guarantees. 
We believe that our region-growing lemma and its application in the context of different
metrics are tools of independent interest that will find further application in the study
of cut problems.  

\smallskip
The hardness proof for \kmic dovetails the hardness proof in~\cite{ChuzhoyMVZ12} for the  
{\em vertex-connectivity} version of \kmic 
(where we want to decrease the $s$-$t$ {\em vertex connectivity} to below $k$), who reduce
from the {\em small-set vertex expansion} (\ssve) problem, which they show is \dks-hard.
We observe that this reduction immediately implies the same hardness for \kmic on a 
{\em directed graph}, 
and combine this with a useful trick from~\cite{ChakrabartyKLN13} that allows us to move
from digraphs to undirected graphs. 
The idea is to take the digraph used in the hardness proof, remove edge directions, and
add some extra nodes and expensive edges 
so that the \emph{residual digraph} obtained after sending a partial $s$-$t$ flow along 
the expensive edges 
essentially coincides with the digraph used in the hardness proof.

\vspace{-1ex}
\paragraph{Related work.}
Standard (i.e., 1-route) cut problems have been extensively studied; we refer the reader
to the textbooks~\cite{AhujaMO93,Vazirani01,WilliamsonS10} for more information.

The study of $k$-route flow and $k$-route cut problems was initiated by
Kishimoto~\cite{Kishimoto96}, and has since received much attention 
in the theoretical Computer Science
community~\cite{BruhnCHKS08,ChekuriK08,BarmanC10,KolmanS11,KolmanS12,ChuzhoyMVZ12}.   
Bruhn et al.~\cite{BruhnCHKS08} gave a $2(k-1)$-approximation for single-source \kmc 
with unit costs, whereas~\cite{ChekuriK08,BarmanC10,KolmanS11} obtained efficient
polylogarithmic approximation results for \kmc with small values of $k$. 
Subsequently, Chuzhoy et al.~\cite{ChuzhoyMVZ12} obtained the first non-trivial results
for \kmc with arbitrary $k$ in the form of bicriteria approximation
guarantees. Independently, Kolman and Scheideler~\cite{KolmanS12} obtained an
$O\bigl(\exp(k)\polylog(r)\bigr)$-approximation for single-source \kmc (with general
costs). As shown by our hardness result for \kmic in Section~\ref{kstcut}, the move
to bicriteria approximations is necessary unless one incurs a $\poly(k)$-factor in the
approximation.  

As noted earlier, $k$-route cut problems 
and {\em network interdiction} problems (see,
e.g.,~\cite{Phillips93,Wood93,Zenklusen10,DinitzG13} and the references therein) can be
viewed as complementary problems. 
For instance, in the {\em maximum-flow interdiction problem} (\mfip) 
we are given {\em edge capacities} in addition to edge costs, and the goal is to minimize
the maximum $s$-$t$ flow in the graph remaining after removing edges of total cost at most
a given budget.  
\mfip with unit capacities is thus complementary to \kmic, and bicriteria guarantees for one
translate to the other. Unit-capacity \mfip is known to be polytime solvable for planar
graphs~\cite{Phillips93,Zenklusen10}. 
Dinitz and Gupta~\cite{DinitzG13} propose a general framework for attacking 
{\em packing interdiction} problems. However, their results do not quite apply to \mfip
(since phrasing max-flow in terms of edge-flows destroys the packing property, and
phrasing it in terms of path-flows yields an interdiction problem where one removes
paths).

\section{\boldmath A simple combinatorial algorithm for $k$-route multiway cut} \label{kmwc}

Recall that in the $k$-route multiway cut (\kmwc) problem, we are given a set
$T=\{t_1,\ldots,t_r\}\sse V$ of terminals and we seek to remove a minimum-cost set of
edges so that the edge-connectivity  
between any two terminals is less than $k$. 
The case $k=1$ is the multiway cut problem, which 
is known to be \apx-hard~\cite{DahlhausJPSY92} even with unit edge costs.
We devise an $\bigl(O(1),O(1)\bigr)$-approximation for \kmwc with unit costs, and an
$\bigl(O(1),O(\log r)\bigr)$-approximation with general costs. 
These improve upon the previous-best guarantees (for general $k$) of
$\bigl(O(1),O(\log^{1.5}r)\bigr)$ for unit costs, and $O\bigl(O(\log r),O(\log^3 r)\bigr)$
for general costs due to~\cite{ChuzhoyMVZ12}.
Remark~\ref{intgap} shows that our guarantees also translate to integrality-gap bounds for
a suitable LP-relaxation. 

Let $\optset$ denote the optimal set of edges, and let $\overline G = (V, E\sm\optset)$ be the
remainder graph. 
Let $k'=k-1$.
Our algorithms are quite easy to describe and analyze.
We first prove a simple claim about $\bG$. 
 
\begin{claim} \label{kmwc-prop}
There is a set $\overline E$ of edges of $\bG$ with $\size{\bE}\leq k'(r-1)$ such that 
$\optset \union \overline{E}$ is a multiway cut in $G$.
\end{claim}

\begin{proof}
Compute a minimum $t_1$-$t_2$ cut $F$ in $\bG$, where $F$ is a set of edges. 
By the definition of $\overline G$, $|F| \leq k'$.
Removing $F$ from $\overline G$ creates at least two components. 
We can now recurse in each connected component, 
and after computing at most $r-1$ min cuts, each terminal will be in a different connected
component. 
\end{proof}

The idea behind the algorithm for unit costs is the following. 
Claim~\ref{kmwc-prop} shows that the optimal multiway cut in $G$ would be a good
approximation to \kmwc if $|\overline{E}|=O\bigl(|\optset|\bigr)$. Otherwise, 
there is a multiway cut in $G$ of cost $O(k'r)$, and so there is some terminal (in fact
$\Omega(r)$ terminals) that has an isolating cut in $G$ of size $O(k')$; we simply remove
this terminal from $T$ and repeat this process. 

\begin{theorem} \label{unitthm}
There is a $\bigl(\gm, \frac{2\gm}{\gm-2}\bigr)$-approximation algorithm for
\kmwc with unit costs for any $\gm>2$. 
\end{theorem}

\begin{proof}
For all $i$, compute a minimum $t_i$-isolating cut $F_i$. 
It is well known that, even with non-unit costs, $\sum_{i=1}^r c(F_i)$ is a
2-approximation to the minimum multiway cut~\cite{DahlhausJPSY92}. In particular,
$C=\sum_{i=1}^r|F_i|\leq 2|\optset|+2|\bE|$ (by Claim~\ref{kmwc-prop}).
If $C\geq\gm k'r$, then we have 
$|\optset|\geq\frac{C}{2}-k'r\geq C\bigl(\frac{1}{2}-\frac{1}{\gm}\bigr)$, 
so taking the union of the $F_i$s yields a $\frac{2\gm}{\gm-2}$-approximation. 
Otherwise, there is some $t_i$ such that $|F_i|<\gm k'$, so we can simply remove $t_i$ 
from $T$ and decrease $r$, and repeat. 

We remark that the number of iterations can be reduced to $\log_2 r$ at
the expense of increasing the connectivity to $2\gm k'$, since must be at least 
$r/2$ terminals such that $|F_i|\leq 2\gm k'$.
\end{proof}

\begin{remark}
The condition $\gm>2$ above is tight. This follows from an example in~\cite{BarmanC10}
where every pair of terminals is $2k'$-edge connected but the minimum multiway cut yields
an $\Omega(r)$-approximation.
\end{remark}

To generalize this algorithm to general edge costs, assume for now that we know
$\OPT=c(\optset)$. Unlike in the unit edge-cost case where we could make progress by
dropping terminals while incurring zero cost, here we will need to 
incur cost $O\bigl(\frac{\OPT}{r}\bigr)$ to drop a terminal 
(or incur cost $O(\OPT)$ to drop $r/2$ terminals). 
This naturally leads to an $O(\log r)$-approximation in the cost. Let
$H_r:=1+\frac{1}{2}+\ldots+\frac{1}{r}=O(\log r)$. 

\begin{theorem} \label{genthm}
There is a $\bigl(\gm,\frac{2\gm}{\gm-2}H_r\bigr)$-approximation algorithm for \kmwc with
general edge costs, for any $\gm>2$. 
\end{theorem}

\begin{proof}
Let $T'$ initialized to $T$ denote the current terminal set and $r'\assign r$.
Let $F$ initialized to $\es$ denote the set of edges removed.
Let $\al=\frac{2}{\gm-2}$.
While $|T'|>1$, we do the following.
Set $c'_e=\min\{c_e,\frac{\al\OPT}{k'r'}\}$ for every edge $e$.
Note that the $c'$-cost of the minimum multiway cut is at most 
$c'(\optset \union \overline{E})\leq\OPT+k'r'\cdot\frac{\al\OPT}{k'r'}$. 
For every terminal $t\in T'$, compute a minimum $c'$-cost $t$-isolating cut $F_t$.
Then, we have $\sum_{t\in T'}c'(F_t)\leq 2(1+\al)\OPT$. So there is some $t\in T'$ such
that $c'(F_t)\leq\frac{2(1+\al)\OPT}{r'}$. The number of edges in $F_t$ with
$c_e>\frac{\al\OPT}{k'r'}$ is less than $\frac{2(1+\al)}{\al}k'=\gm k'$. 
We add edges in $F_t$ with $c_e\leq\frac{\al\OPT}{k'r'}$ to $F$. 
This incurs cost at most
$c'(F_t)\leq\frac{2(1+\al)\OPT}{r'}$, and ensures that $t$ is less than $\gm k'$
connected to every other terminal in $T'$ in the remaining graph. 
We now set $T'\assign T'\sm\{t\},\ r'\assign r'-1$, and repeat the above process.

Clearly, every pair of terminals is at most $\gm k'$ connected in $(V,E\sm F)$. Also,
$c(F)\leq\sum_{r'=r}^1\frac{2(1+\al)\OPT}{r'}=\OPT\cdot\frac{2\gm}{\gm-2}\cdot H_{r}$.

Finally, we can eliminate the need for knowing $\OPT$ as follows.
Given a guess $C$ of $\OPT$, if at some iteration we have
$\sum_{t\in T'}c'(F_t)>2(1+\al)C$ then we know that $C<\OPT$; otherwise, we obtain a
solution of cost at most $C\cdot\frac{2\gm}{\gm-2}\cdot H_r$. So we can try powers of
$(1+\e)$ to find the smallest $C$ such that the latter case happens; this blows up the
approximation in cost by at most a $(1+\e)$-factor.
\end{proof}

\begin{remark}[LP-relative bounds] \label{intgap}
The guarantees in Theorems~\ref{unitthm} and~\ref{genthm} also translate to
integrality-gap bounds for the following LP-relaxation of \kmwc.
Let $\Pc_{ij}$ be the collection of all $t_i$-$t_j$ paths. 
\begin{equation}
\min \ \ c^Tx \quad \ \ \text{s.t.} \qquad 
\sum_{e\in P}(x_e+y_e)\geq 1 \quad \forall t_i,t_j\in T,\ P\in\Pc_{ij}; \quad \ \ 
\sum_e y_e\leq k'(r-1); \quad \ \ x,y \geq 0. \tag{P'} \label{kmwclp}
\end{equation}
\newcommand{\optkmwclp}{\OPT_{\text{\ref{kmwclp}}}}
Claim~\ref{kmwc-prop} implies that \eqref{kmwclp} is indeed a valid relaxation of \kmwc.
Let $(x,y)$ be an optimal solution to \eqref{kmwclp} and $\optkmwclp$ be its value. 
Then, for any $\ld\geq 0$,
for the cost function $c'_e=\min\{c_e,\ld\}$, there is a fractional multiway-cut of
$c'$-cost at most $\optkmwclp+\ld\sum_ey_e$. Also, if $F_i$ is a minimum $c'$-cost
$t_i$-isolating cut then we have $\sum_ic'(F_i)\leq 2\bigl(\optkmwclp+\ld\sum_e y_e\bigr)$. 
(This follows since an optimal solution to the multiway-cut LP is known to be
half-integral (see, e.g.,~\cite{Vazirani01}); this implies that 
2(cost of an optimal solution) is at least 
$\sum_{i}$ (cost of a minimum $t_i$-isolating cut).)
This implies that we can replace $|\optset|$ and $|\bE|$ in the proof of
Theorem~\ref{unitthm} by $\optkmwclp$ and $\sum_e y_e$, and $\OPT$ in the proof of
Theorem~\ref{genthm} by $\optkmwclp$, and all the arguments go through. 
\end{remark}

\vspace{-1ex}
\paragraph{The all-pairs case.}
This is the special case of \kmwc where $T=V$. 
To our knowledge, this {\em $k$-route all-pairs cut} problem has not been explicitly
studied before.  
When $k=1$, the all-pairs problem is trivial as the remainder graph cannot contain any
edge. When $k=2$, this problem is still in \p as the remainder graph is a maximum-cost
spanning forest. 
We prove that the problem is \apx-hard for all $k\geq 3$ (see
Appendix~\ref{append-allpairs}), thus resolving the complexity (with respect to polytime
solvability) of $k$-route all-pairs cut.  
The all-pairs problem can also be stated in terms of properties required of the
remainder graph.  
For example, in $3$-route all-pairs cut, we seek a minimum-cost edge set such
that the remainder graph does not contain a {\em diamond} as a minor. Interestingly, this
is equivalent to requiring that the remainder graph be a maximum-weight \emph{cactus},
which is a graph where every edge lies in at most one cycle. As noted above, this problem
is \apx-hard. But we observe that this problem admits an $O(1)$-approximation as a
consequence of the results of Fiorini et al.~\cite{FioriniJP10}; see
Appendix~\ref{append-allpairs}.

\section{\boldmath A region-growing algorithm for $k$-route multicut} \label{kmc} 
We now consider general $k$-route multicut (\kmc): given source-sink pairs/commodities 
$(s_1,t_1),\ldots,(s_r,t_r)$, we want to find a minimum-cost set of edges whose
removal reduces the $s_i$-$t_i$ edge connectivity to less than $k$ for all $i=1,\ldots,r$.
We consider the following LP-relaxation of the problem, which was also considered by
Barman and Chawla~\cite{BarmanC10}. Throughout $e$ indexes the edges in $E$, and $i$ indexes
the commodities. Let $\Pc_i$ denote the collection of all (simple) $s_i$-$t_i$ paths in
$G$. 
\begin{equation}
\min \ \ \sum_e c_ex_e \quad \ \ \text{s.t.} \quad \ \ 
\sum_{e\in P} (x_e+y^i_e) \geq 1 \quad \forall i, \forall P\in\Pc_i; \quad \ \ 
\sum_e y^i_e \leq k-1 \quad \forall i; \quad \ \ 
x,y \geq 0.
\tag{P} \label{kmcp}
\end{equation}

Let $\bigl(x,\{y^i\}\bigr)$ denote an optimal solution to \eqref{kmcp}, and $\OPT$ be its
value. 
We show that 
$\bigl(x,\{y^i\}\bigr)$ can
be rounded to yield an $O(\ln r\ln\ln r)$-approximation when $k=2$, and a bicriteria 
$\bigl(\gm,O\bigl(\frac{\gm}{(\sqrt{\gm}-1)^2}\ln r\ln\ln r\bigr)\bigr)$-approximation
for \kmc with unit edge costs, for any $\gm>1$.
Notably, our cost-approximation is with respect to $\OPT$, and so they translate to 
integrality-gap upper bounds for \eqref{kmcp}.  
Our results improve upon (for any fixed $\gm$) the previous-best guarantees for these
cases in~\cite{ChuzhoyMVZ12} by roughly a $\sqrt{\log r}$-factor. 

\subsection{Region-growing lemmas} \label{reglemmas}
The main tool that we leverage is 
{\em region growing}~\cite{LeightonR99,GargVY96,EvenNRS00}.  
The idea is to view the LP
solution as a metric, grow a suitable ball in this metric and prove 
that the cost of the ball-boundary edges can be charged to the ball-volume, where
volume measures the contribution to the LP objective from the edges 
inside the ball. 
The main difficulty in applying this idea to \eqref{kmcp} is that, unlike
most applications of region growing~\cite{LeightonR99,GargVY96,EvenNSS98,EvenNRS00}, the
LP solution yields a {\em different metric} for each commodity. 
The key technical ingredient and novelty is a new region-growing lemma (see
Lemmas~\ref{reggrow} and~\ref{cor2}) 
that is analogous to, but more general, than the one in~\cite{EvenNRS00}, 
and much more sophisticated than the one used in~\cite{BarmanC10,KolmanS11,KolmanS12}.
Roughly speaking, we prove that given a current set $S$ of nodes, one can construct a ball
around any $s_i$ in the $(x+y^i)$-metric such that the cost of the ``boundary $x$-edges''
can be charged to $(x\text{-volume of $S$})\cdot
\ln\bigl(\text{$x$-volume of $S$}/\text{$x$-volume of the ball}\bigr)\cdot\ln\ln r$
(Lemma~\ref{cor2}). A subtle insight that helps deal with the complication that 
different applications of region growing involve different commodities and therefore
different metrics is the following. 
Since the $x$-contribution is {\em common} to all $(x+y^i)$-metrics, even though 
we consider different commodity-metrics we 
can leverage the above guarantee 
and obtain the same kind of savings that~\cite{EvenNRS00} obtain in their
divide-and-conquer algorithms (see Lemma~\ref{indnlem}); this leads to our improved
approximation guarantees. 

We now state our region-growing lemmas in a general form and then apply these 
to the optimal solution $\bigl(x,\{y^i\}\bigr)$ to obtain various useful corollaries.
In Section~\ref{node-extn}, we extend our arguments to prove region-growing lemmas in
settings that involve both edge and node lengths.
Let $n=|V|$, $m=|E|$.
Let $\ell:V\times V\mapsto\R_{\geq 0}$ be a metric on $V\times V$. 
Our algorithm will iteratively focus on certain regions of the graph $G$. 
Let $S\sse V$, which is intended to represent the node-set of the current region. 
Let $F\sse E$, which is intended to represent the edges that contribute to the volume, and
whose cost we care about.
Let $\beta\geq 0$. 
Let $z\in V$, and $\rho\geq 0$.

\begin{list}{$\bullet$}{\itemsep=0ex \topsep=0.5ex \addtolength{\leftmargin}{-2ex}} 
\item Define $B_\ell(z,\rho):=\{v\in V: \ell_{zv}\leq\rho\}$ to be the ball of radius
$\rho$ around $z$. 

\item Let $B^S_\ell(z,\rho):=B_\ell(z,\rho)\cap S$ and 
$\ov{B^S_\ell}(z,\rho):=S\sm B_\ell(z,\rho)$.

\item Define the following volumes:
\begin{eqnarray*}
\V^{S,F}_\ell(\beta; z,\rho) & := & \beta
+\sum_{(u,v)\in F: u,v\in B^S_\ell(z,\rho)}c_{uv}\ell_{uv}
+\sum_{\substack{(u,v)\in F: u\in B^S_\ell(z,\rho) \\ v\in\ov{B^S_\ell(z,\rho)}}}c_{uv}\bigl(\rho-\ell_{zu}\bigr) \\[0.5ex]
\ov{\V^{S,F}_\ell}(\beta; z,\rho) & := & \beta
+\sum_{(u,v)\in F: u,v\in\ov{B^S_\ell(z,\rho)}}c_{uv}\ell_{uv}
+\sum_{\substack{(u,v)\in F: u\in B^S_\ell(z,\rho) \\ v\in\ov{B^S_\ell(z,\rho)}}}c_{uv}\bigl(\ell_{zv}-\rho\bigr)
\end{eqnarray*}

\item For a subset $T\sse S$, let $\dt^S_F(T)$ denote $\{(u,v)\in F: u\in T, v\in S\sm T\}$.

\item Define
$\del^{S,F}_\ell(z,\rho):=\dt^S_F\bigl(B^S_\ell(z,\rho)\bigr)=\dt^S_F\bigl(\ov{B^S_\ell}(z,\rho)\bigr)$. 
\end{list}

When $F=E$, we drop $F$ from the above pieces of notation (e.g., $\dt^S_E(T)$
is shortened to $\dt^S(T)$).
For $H\sse E$, we use $\ell(H)$ to denote $\sum_{e\in H}\ell_e$.

\begin{lemma}[Region-growing lemma] \label{reggrow}
Let $F\sse E$, $S\sse V$, $z\in V$, and $0\leq a<b$.
Let $\rho$ be chosen uniformly at random from $[a,b)$. Then,
\begin{eqnarray}
\E[\rho]{\frac{c\bigl(\del^{S,F}_\ell(z,\rho)\bigr)}
{\V^{S,F}_\ell(\beta;z,\rho)\ln\Bigl(\frac{e\V^{S,F}_\ell(\beta;z,b)}{\V^{S,F}_\ell(\beta;z,\rho)}\Bigr)}}
& \leq & \frac{1}{b-a}\cdot\ln\ln\biggl(\frac{e\V^{S,F}_\ell(\beta;z,b)}{\V^{S,F}_\ell(\beta;z,a)}\biggr) 
\quad\text{and} \label{regbnd1} \\
\E[\rho]{\frac{c\bigl(\del^{S,F}_\ell(z,\rho)\bigr)}
{\ov{\V^{S,F}_\ell}(\beta;z,\rho)\ln\Bigl(\frac{e\ov{\V^{S,F}_\ell}(\beta;z,a)}{\ov{\V^{S,F}_\ell}(\beta;z,\rho)}\Bigr)}}
& \leq &
\frac{1}{b-a}\cdot\ln\ln\biggl(\frac{e\ov{\V^{S,F}_\ell}(\beta;z,a)}{\ov{\V^{S,F}_\ell}(\beta;z,b)}\biggr) 
\label{regbnd2}
\end{eqnarray}
\end{lemma}

\begin{proof} 
We abbreviate $c\bigl(\del^{S,F}_\ell(z,\rho)\bigr)$ to $c(\rho)$,
$\V^{S,F}_\ell(\beta;z,\rho)$ to $\V(\rho)$ and $\ov{\V^{S,F}_\ell}(\beta;z,\rho)$ to
$\ov{\V}(\rho)$.
Let $x^-$ be a value infinitesimally smaller than $x$.
Let $I=\{\ell(s_i,v): v\in V\}$. Note that $\V(\rho)$ and
$\ov{\V}(\rho)$ are differentiable at all $\rho\in[a,b)\sm I$ and for
each such $\rho$, we have
$\frac{d\V(\rho)}{d\rho}=c(\rho)$ and 
$\frac{d\ov{\V}(\rho)}{d\rho}=-c(\rho)$. 
Let $a_0=a$, $a_k=b$, and $\{a_1,\ldots,a_{k-1}\}=(a,b)\cap I$ with $a_1<\ldots<a_{k-1}$. 
Then,
\begin{alignat*}{2}
(b-a)&\cdot\E[\rho]{\frac{c(\rho)}{\V(\rho)\ln\bigl(\frac{e\V(b)}{\V(\rho)}\bigr)}}
&& =  
\sum_{i=1}^k\int_{a_{i-1}}^{a_i^-}\frac{d\V(\rho)}{\V(\rho)\ln\bigl(\frac{e\V(b)}{\V(\rho)}\bigr)}
=
\sum_{i=1}^k\int_{a_{i-1}}^{a_i^-}\frac{d(\ln\V(\rho))}{\ln\bigl(e\V(b)\bigr)-\ln(\V(\rho))} \\
& = \sum_{i=1}^k-\ln\ln\Bigl(\frac{e\V(b)}{\V(\rho)}\Bigr)\biggr|_{a_{i-1}}^{a_i^-} 
&& = \sum_{i=1}^k
\biggl[\ln\ln\Bigl(\frac{e\V(b)}{\V(a_{i-1})}\Bigr)-\ln\ln\Bigl(\frac{e\V(b)}{\V(a_i^-)}\Bigr)\biggr]
\leq \ln\ln\Bigl(\frac{e\V(b)}{\V(a)}\Bigr).
\end{alignat*}
The final inequality follows since $\V(\rho)$ increases with $\rho$, and
$\ln\ln\bigl(\frac{e\V(b)}{\V(\rho)}\bigr)$ decreases with $\rho$.
Similarly, we obtain that 
$
(b-a)\cdot\E[\rho]{\frac{c(\rho)}{\ov{\V}(\rho)\ln\bigl(\frac{e\ov{\V}(a)}{\ov{\V}(\rho)}\bigr)}}
= \sum_{i=1}^k\ln\ln\Bigl(\frac{e\ov{\V}(a)}{\ov{\V}(\rho)}\Bigr)\biggr|_{a_{i-1}}^{a_i^-} 
\leq \ln\ln\Bigl(\frac{e\ov{\V}(a)}{\ov{\V}(b)}\Bigr)$. 
\end{proof}

\begin{corollary} \label{cor1}
Let $F, H \sse E$, $S\sse V$, $z\in V$, and $0\leq a<b$. 
For any $\al\in(0,1)$, we can efficiently find a radius $\rho_1\in[a,b)$ such that
\begin{eqnarray}
c\bigl(\del^{S,F}_\ell(z,\rho_1)\bigr)
& \leq & \frac{2}{(1-\al)(b-a)}
\cdot\V^{S,F}_\ell(\beta;z,\rho_1)
\cdot\ln\Bigl(\tfrac{e\V^{S,F}_\ell(\beta;z,b)}{\V^{S,F}_\ell(\beta;z,\rho_1)}\Bigr)
\cdot\ln\ln\Bigl(\tfrac{e\V^{S,F}_\ell(\beta;z,b)}{\V^{S,F}_\ell(\beta;z,a)}\Bigr)
\label{fbound1} \\
c\bigl(\del^{S,F}_\ell(z,\rho_1)\bigr)
& \leq & \frac{2}{(1-\al)(b-a)}
\cdot\ov{\V^{S,F}_\ell}(\beta;z,\rho_1)
\cdot\ln\Bigl(\tfrac{e\ov{\V^{S,F}_\ell}(\beta;z,a)}{\ov{\V^{S,F}_\ell}(\beta;z,\rho_1)}\Bigr)
\cdot\ln\ln\Bigl(\tfrac{e\ov{\V^{S,F}_\ell}(\beta;z,a)}{\ov{\V^{S,F}_\ell}(\beta;z,b)}\Bigr)
\label{fbound2} \\
\bigl|\del^{S,H}_\ell(z,\rho_1)\bigr| & < & \frac{\ell(H)}{\al(b-a)}. \label{hbound2}
\end{eqnarray}
\end{corollary}

\begin{proof} 
Suppose we pick $\rho$ uniformly at random from $[a,b)$. 
Define the following events.
\begin{eqnarray*}
\Om & := & \Bigl\{\rho\in[a,b):
\bigl|\del^{S,H}_\ell(z,\rho)\bigr|\geq\frac{\ell(H)}{\al(b-a)}\Bigr\} \\
\Om_1 & := & \biggl\{\rho\in[a,b): c\bigl(\del^{S,F}_\ell(z,\rho)\bigr)>
\frac{2}{(1-\al)(b-a)}\cdot 
\V^{S,F}_\ell(\beta;z,\rho)\ln\Bigl(\tfrac{e\V^{S,F}_\ell(\beta;z,b)}{\V^{S,F}_\ell(\beta;z,\rho)}\Bigr)
\cdot\ln\ln\Bigl(\tfrac{e\V^{S,F}_\ell(\beta;z,b)}{\V^{S,F}_\ell(\beta;z,a)}\Bigr)\biggr\}
\\
\Om_2 & := & \biggl\{\rho\in[a,b): c\bigl(\del^{S,F}_\ell(z,\rho)\bigr)>
\frac{2}{(1-\al)(b-a)}\cdot 
\ov{\V^{S,F}_\ell}(\beta;z,\rho)\ln\Bigl(\tfrac{e\ov{\V^{S,F}_\ell}(\beta;z,a)}{\ov{\V^{S,F}_\ell}(\beta;z,\rho)}\Bigr)
\cdot\ln\ln\Bigl(\tfrac{e\ov{\V^{S,F}_\ell}(\beta;z,a)}{\ov{\V^{S,F}_\ell}(\beta;z,b)}\Bigr)\biggr\}
\end{eqnarray*}

For an edge $(u,v)\in E$,
$\Pr\bigl[(u,v)\in\del^{S,H}_\ell(z,\rho)]\leq\frac{\ell_{uv}}{b-a}$. 
Hence,
$\E[\rho]{\bigl|\del^{S,H}_\ell(z,\rho)\bigr|}\leq\frac{\ell(H)}{b-a}$, 
and therefore $\Pr[\Om]\leq\al$.
By Lemma~\ref{reggrow} and Markov's inequality, we have that 
$\Pr[\Om_1],\Pr[\Om_2]<(1-\al)/2$. Conditioning on $\Om^c:=[a,b)\sm\Om$ increases the
probability of an event by at most a factor $\frac{1}{1-\Pr[\Om]}\leq\frac{1}{1-\al}$, so 
$\Pr[\Om_1|\Om^c],\Pr[\Om_2|\Om^c]<1/2$.
Thus, $\Pr[\Om^c\cap\Om_1^c\cap\Om_2^c]>0$.

We argue that $\Om^c$, $\Om_1^c$, and $\Om_2^c$ are all unions of at most $n$ subintervals
of $[a,b)$, and we can find these efficiently.
Since $\Pr[\Om^c\cap\Om_1^c\cap\Om_2^c]>0$, we can then efficiently find an
interval contained in $\Om^c\cap\Om_1^c\cap\Om_2^c$ (in fact, a non-singleton interval),  
and hence, find $\rho_1\in[a,b]$ satisfying \eqref{fbound1}--\eqref{hbound2} (in fact,
there are infinitely many such $\rho$).  
 
There are at most $n$ distinct sets $B^S_\ell(z,\rho)$ that one may encounter as
$\rho$ varies in $[a,b)$. For each such set $A$, there is an interval $[\lb,\ub)$ such
that $A=B^S_\ell(z,\rho)$ for all $\rho\in[\lb,\ub)$. Note that the right-hand-sides
(RHSs) of \eqref{fbound1} and \eqref{fbound2} are continuous and differentiable in
$(\lb,\ub)$, and are monotonic (increasing and decreasing, respectively) functions of
$\rho$. We call $[\lb,\ub)$ a smooth subinterval of $[a,b)$. By definition, the
left-hand-sides (LHSs) of \eqref{fbound1}--\eqref{hbound2} are invariant over a smooth
subinterval. Hence, $\Om^c$ is the union of some smooth subintervals. Consider a smooth
subinterval $[\lb,\ub)$.  
By continuity, if some $\rho\in[\lb,\ub)$ satisfies \eqref{fbound1} or 
\eqref{fbound2}, then we can efficiently find the maximal subinterval of $[\lb,\ub)$
(which may be a singleton interval) such that all $\rho$ in the subinterval satisfy
the given bound.  
Hence, both $\Om_1^c$ and $\Om_2^c$ are the union of at most $n$ subintervals of $[a,b)$. 
By trying out the at most $3n$ possible subintervals of $\Om_1^c\cup\Om_1^c\cup\Om_2^c$,
we can find some interval contained in 
$\Om^c\cap\Om_!^c\cap\Om_2^c$ 
and hence obtain $\rho_1$ satisfying \eqref{fbound1}--\eqref{hbound2}. 
\end{proof}

\vspace{-1ex}
\paragraph{Applications of the region-growing lemmas.} 
To apply the above results to the metrics obtained from 
$\bigl(x,\{y^i\}\bigr)$, 
it will be convenient to modify $G$ by subdividing every edge $e$ into $r+1$ edges
$e_0,e_1,\ldots,e_r$, and setting $x_{f}=x_e$ for $f=e_0$ and 0 otherwise, and
$y^i_f=y^i_e$ if $f=e_i$ and 0 otherwise. We call $e_0$ an $x$-edge, and we call $e_i$, a
$y^i$-edge for $i=1,\ldots,r$. 
Clearly, any solution in $G$ yields a solution in the subdivided graph of the same cost
and vice versa, and this holds even for fractional solutions to \eqref{kmcp}. 
In the sequel, we work with the subdivided graph.  
To keep notation simple, we continue to use $G=(V,E)$ to denote the subdivided graph, and
$\bigl(x,\{y^i\}\bigr)$ to denote the above solution in the subdivided graph. 
Let $F$ be the set of all $x$-edges, and $H^i$ be the set of all $y^i$-edges for all
$i=1,\ldots,r$. 
Consider a commodity $i$.
Let $\ell^i$ denote the shortest-path metric of $G$ (i.e., the subdivided graph) induced
by the $\{x_e+y^i_e\}$ edge lengths. Set $\beta=\OPT/r$.
To avoid cumbersome notation, we shorten: 

\vspace{0.5ex}
\begin{tabular}{l@{\qquad to \qquad }l}
$B^S_{\ell^i}(z,\rho)$ and $\ov{B^S_{\ell^i}}(z,\rho)$ & 
$B^S_i(z,\rho)$ and $\ov{B^S_i}(z,\rho)$ respectively \\[0.75ex]
$\V^{S,F}_{\ell^i}(\beta; z,\rho)$ 
and $\ov{\V^{S,F}_{\ell^i}}(\beta; z,\rho)$ & 
$\V^{S,x}_i(z,\rho)$ and $\ov{\V^{S,x}_i}(z,\rho)$ respectively \\[0.5ex]
$\del^{S,F}_{\ell^i}(z,\rho)$, $\del^{S,H^i}_{\ell^i}(z,\rho)$, and 
$\del^{S,E}_{\ell^i}(z,\rho)$ &  
$\del^{S,x}_i(z,\rho)$, $\del^{S,y}_i(z,\rho)$, and $\del^S_i(z,\rho)$ respectively \\[1.25ex]
\end{tabular}

\noindent
Also, define $\V^x(S):=\beta+\sum_{e\in E(S)}c_ex_e$, where $E(S)$ is the set of edges having
both endpoints in $S$.
Finally, for an integer $q\geq 1$ and a set $A$ of edges let $c^q(A)$ be the cost of all
but the $q-1$ most expensive edges of $A$ (so $c^q(A)=0$ if $|A|<q$).

\begin{lemma} \label{cor2}
Let $S\sse V$, $z\in V$, and $i$ be some commodity. 
Let $\al\in(0,1)$ and $q=\ceil{\frac{k-1}{\al}}$.
We can efficiently find $\rho_1\in[0,1)$ such that
\begin{eqnarray}
c^q\bigl(\del^S_i(z,\rho_1)\bigr) & \leq &
\frac{2}{1-\al}\cdot\V^{S,x}_i(z,\rho_1)\ln\Bigl(\frac{e\V^{x}(S)}{\V^{S,x}_i(z,\rho_1)}\Bigr)%
\ln\ln\bigl(e(r+1)\bigr) \label{xbound1} \\
c^q\bigl(\del^S_i(z,\rho_1)\bigr) & \leq &
\frac{2}{1-\al}\cdot\ov{\V^{S,x}_i}(z,\rho_1)\ln\Bigl(\frac{e\V^{x}(S)}{\ov{\V^{S,x}_i}(z,\rho_1)}\Bigr)%
\ln\ln\bigl(e(r+1)\bigr). \label{xbound2}
\end{eqnarray}
\end{lemma}

\begin{proof}
We apply Corollary~\ref{cor1} taking $\ell=\ell^i$, $H=H^i$, and $[a,b)=[0,1)$ (and $S$,
$z$, $\al$ as given by the statement of the lemma, and $F$ to be the set of $x$-edges).
Note that 
\vspace{-1ex}
$$
\V^{S,x}_i(z,1)\leq\V^x(S), \qquad \ov{\V^{S,x}_i}(z,0)\leq\V^x(S), \qquad
\frac{\V^{S,x}_i(z,1)}{\V^{S,x}_i(z,0)}\leq r+1, \qquad
\frac{\ov{\V^{S,x}_i}(z,0)}{\ov{\V^{S,x}_i}(z,1)}\leq r+1.
$$

\vspace{-1ex}
Thus, we obtain $0\leq\rho_1<1$ such that
$c\bigl(\del^{S,x}_i(z,\rho_1)\bigr)$ 
satisfies the bounds given by the RHS of \eqref{xbound1} and \eqref{xbound2} respectively.  
Moreover, since $\ell^i(H^i)\leq k-1$, we have that 
$\bigl|\del^{S,y}_i(z,\rho_1)\bigr|<\frac{k-1}{\al}$ due to \eqref{hbound2}, 
and so $\bigl|\del^{S,y}_i(z,\rho_1)\bigr|<q$.
Finally, since edges not in $F\cup H^i$ have zero $\ell^i$-length, 
$\del^{S,x}_i(z,\rho)$ and $\del^{S,y}_i(z,\rho)$ partition $\del^S_i(z,\rho)$ for all $\rho$. 
Therefore, 
$c\bigl(\del^{S,x}_i(z,\rho_1)\bigr)\geq c^q\bigl(\del^S_i(z,\rho_1)\bigr)$, 
and the lemma follows.
\end{proof}

\begin{corollary} \label{cor3}
Let $S\sse V$. Suppose that $s_i,t_i\in S$ and there are $\gm(k-1)$ edge-disjoint
$s_i$-$t_i$ paths internal to $S$, 
where $\gm>1$. Suppose that $c_e=1$ for all edges $e$.  
We can efficiently find 
$\rho_1\in[0,1)$ such that
\begin{eqnarray*}
c\bigl(\del^S_i(s_i,\rho_1)\bigr)\leq
\frac{2\gm}{(\sqrt{\gm}-1)^2}\cdot\V^{S,x}_i(s_i,\rho_1)\ln\Bigl(\frac{e\V^x(S)}{\V^{S,x}_i(s_i,\rho_1)}\Bigr)%
\ln\ln\bigl(e(r+1)\bigr) \\
c\bigl(\del^S_i(s_i,\rho_1)\bigr)\leq
\frac{2\gm}{(\sqrt{\gm}-1)^2}\cdot\ov{\V^{S,x}_i}(s_i,\rho_1)\ln\Bigl(\frac{e\V^x(S)}{\ov{\V^{S,x}_i}(s_i,\rho_1)}\Bigr)%
\ln\ln\bigl(e(r+1)\bigr)
\end{eqnarray*}
\end{corollary}

\begin{proof} 
We apply Lemma~\ref{cor2} with $\al\in(0,1)$, whose value we will fix later, to obtain
$\rho_1\in[0,1)$. Note that $B^S_i(s_i,\rho_1)$ is an $s_i$-$t_i$ cut. 
Let $q=\ceil{\frac{k-1}{\al}}$.
For any $s_i$-$t_i$ cut $A\sse S$, we know that $c\bigl(\dt^S(A)\bigr)\geq\gm(k-1)$.
Therefore, since we have unit edge costs, 
$c\bigl(\dt^S(A)\bigr)\leq c^q\bigl(\dt^S(A)\bigr)+q-1
\leq c^q\bigl(\dt^S(A)\bigr)\cdot\frac{\gm}{\gm-1/\al}$.
Plugging in the bounds for $c^q(.)$ from \eqref{xbound1}, \eqref{xbound2}, we see that the
constant factor multiplying the volume terms on the RHS is
$\frac{2\gm}{(1-\al)(\gm-1/\al)}$. This factor is minimized by setting $\al=\gm^{-1/2}$,
which yields the constant factor $\frac{2\gm}{(\sqrt{\gm}-1)^2}$ and completes the proof.
\end{proof}

\subsection{The rounding algorithms and their analyses} \label{kmcalg}

\paragraph{\boldmath The case $k=2$.}
The algorithm for $k=2$ follows a similar template as the algorithm in~\cite{BarmanC10} for
2-\mc. However, its analysis resulting in our improved guarantee relies crucially on
Lemma~\ref{cor2} which is derived from our stronger region-growing lemma.  
The algorithm proceeds as follows. 
Given a current set $U$ of nodes 
and a current set of $N$ source-sink pairs, 
we repeatedly use Lemma~\ref{cor2} to ``carve out'' disjoint regions
$A_1,\ldots,A_h\sse U$ and build a set $Z$ of edges until there is no $2$-edge-connected
source-sink pair in $U\sm(A_1\cup\ldots\cup A_h)$. 
Given $A_1,\ldots,A_{p-1}$, we obtain $A_p$ as follows. We choose an $s_i$-$t_i$ pair that 
is at least $2$-edge connected in $S=U\sm\bigcup_{q=1}^{p-1}A_q$, and use Lemma~\ref{cor2}
with center $s_i$, $\al=0.5$ and set $S$. 
Note that 
$2(k-1)=2=k$. 
We set $A_p$ to be $B^S_i(s_i,\rho_1)$ or $\ov{B^S_i}(s_i,\rho_1)$, so as to ensure there
are at most $N/2$ source-sink pairs inside $A_p$. 
We add the edges corresponding to $c^2\bigl(\dt^S(A_p)\bigr)$ to $Z$; Lemma~\ref{cor2}
ensures that the cost of these edges can be bounded in terms of the volume contained in
$A_p$. Having obtained $A_1,\ldots,A_h$ this way, we now recurse on each set $A_p$ and the
source-sink pairs contained in $A_p$, to obtain edge-sets $Z_1,\ldots,Z_h$. The solution
we return is $Z\cup(Z_1\cup\ldots\cup Z_h)$. 
A more formal description follows. 

\vspace{1.5ex}

{\small \hrule \vspace{5pt}
\noindent {\bf Algorithm} $\twomcalg(U, \T=\{(s_1,t_1),\ldots,(s_{N},t_{N})\}$) 

\noindent {\bf Input:} A subset $U\sse V$, and a collection $\T=\{(s_i,t_i)\}_{i=1}^{N}$ of
source-sink pairs, where $s_i,t_i\in U$ for all $i=1,\ldots, N$.

\noindent {\bf Output:} A set $Z\sse E(U)$ such that $s_i$ and $t_i$ are at most
1-edge-connected in $(U,E(U)\sm Z)$ for all $i=1,\ldots,N$.

\begin{list}{A\arabic{enumi}.}{\usecounter{enumi} \topsep=0.5ex \itemsep=0.25ex \addtolength{\leftmargin}{-1.5ex}}
\item Set $S=U$, $Z=\es$, $\Sc=\es$, and
$\T'=\bigl\{(s_i,t_i)\in\T: \text{$s_i$ and $t_i$ are at least $2$-edge-connected in $(S,E(S))$}\bigr\}$. 

\item If $\T'=\es$, return $Z$.

\item While $\T'\neq\es$, we do the following.
\begin{list}{A\arabic{enumi}.\arabic{enumii}}{\topsep=0.25ex \itemsep=0ex \usecounter{enumii}}
\item Pick some $(s_i,t_i)\in\T'$.
\item Apply Lemma~\ref{cor2} with $z=s_i$, $\al=0.5$ and the set $S$ to find a radius
$0\leq\rho_1<1$ satisfying \eqref{xbound1}, \eqref{xbound2}.
\item If $B^S_i(s_i,\rho_1)$ contains at most $N/2$ pairs from $\T$ then set
$A=B^S_i(s_i,\rho_1)$, else set $A=\ov{B^S_i}(s_i,\rho_1)$. 
\item Set $\Sc\assign\Sc\cup\{A\}$. 
Add the edges contributing to $c^2\bigl(\dt^S(A)\bigr)$ (i.e., all edges of $\dt^S(A)$
except the most-expensive one) to $Z$.  
\item  Set $S\assign S\sm A$. 
Update $\T'$ to be the $s_i$-$t_i$ pairs from $\T$ that are at least $2$-edge-connected in
$(S,E(S))$.
\end{list}

\item For every set $A\in\Sc$, set 
$Z\assign Z\cup\twomcalg\bigl(A,\{(s_i,t_i)\in\T: s_i,t_i\in A\}\bigr)$.

\item Return $Z$.
\end{list}

\noindent The initial call to $\twomcalg$, which computes the solution we return, is
$\twomcalg\bigl(V,\{(s_1,t_1),\ldots,(s_r,t_r)\}\bigr)$. 
\hrule
}

\vspace{1.5ex}
Let $Z:=\twomcalg\bigl(V,\{(s_1,t_1),\ldots,(s_r,t_r)\}\bigr)$. 
The feasibility of $Z$ follows from the same arguments as in~\cite{BarmanC10};
Lemma~\ref{2feas} gives a 
self-contained proof. 
We focus on showing that $c(Z)\leq O(\ln r\ln\ln r)\cdot\OPT$. 
Consider the recursion tree associated with the execution of \twomcalg, where each node is
labeled with arguments passed in the current invocation of $\twomcalg$. 
Define the depth of a subtree of this recursion tree to be the maximum number of edges on
a root to leaf path of the subtree.
Recall that $\beta=\frac{\OPT}{r}$.

\begin{lemma} \label{indnlem}
Let $d$ be the depth of a subtree of the recursion tree rooted at 
$\bigl(U\sse V,\T\sse\{(s_1,t_1),\ldots,(s_r,t_r)\}\bigr)$. 
Let $Z_U=\twomcalg(U,\T)$. 
We have 
$c(Z_U)\leq 
4\bigl(\beta|\T|+\V^x(U)\bigr)\ln\bigl(\frac{e^dr\V^x(U)}{\OPT}\bigr)\ln\ln\bigl(e(r+1)\bigr)$.
\end{lemma}

\begin{proof}
The proof is by induction on $d$. If $d=0$, there is no recursive call; so $Z^U=\es$,
which satisfies the stated bound.
Otherwise, suppose that we make the recursive calls 
$\twomcalg(A_1,\T_1),\ldots,\twomcalg(A_h,\T_h)$ in step A4 to obtain edge-sets
$Z_1,\ldots,Z_h$ respectively. 
For $p=1,\ldots,h$, let $S_p$ be the current set $S$ when set $A_p$ was added to $\Sc$ in
step A3.4 (so $A_p\sse S_p$), and let $E_p$ be the edge-set added to $Z$ in this step.
Then, $Z_U=\bigcup_{p=1}^h(E_p\cup Z_p)$. 
By the induction hypothesis, 
$c(Z_p)\leq 4\Bigl(\beta|\T_p|+\V^x(A_p)\Bigr)\ln\Bigl(\frac{e^{d-1}r\V^x(A_p)}{\OPT}\Bigr)\ln\ln\bigl(e(r+1)\bigr)$.
Let $\vol_p=\V^{S_p,x}_i(s_i,\rho_1)$ if $A_p=B^{S_p}_i(s_i,\rho_1)$ 
and $\vol_p=\ov{\V^{S_p,x}_i}(s_i,\rho_1)$ if $A_p=\ov{B^{S_p}_i}(s_i,\rho_1)$.
Note that $\V^x(A_p)\leq\vol_p\leq\V^x(S_p)\leq\V^x(U)$.
By Lemma~\ref{cor2} and the above upper bounds, we have
\begin{eqnarray}
c(E_p) & = & c^2\bigl(\dt^{S_p}(A_p)\bigr) \leq 
4\vol_p\ln\Bigl(\frac{e\V^x(U)}{\V^{x}(A_p)}\Bigr)\ln\ln\bigl(e(r+1)\bigr), 
\qquad \text{and therefore} \notag \\
c(E_p)+c(Z_p) & \leq &
4\Bigl(\beta|\T_p|+\vol_p\Bigr)\ln\Bigl(\frac{e^dr\V^x(U)}{\OPT}\Bigr)\ln\ln\bigl(e(r+1)\bigr). 
\label{indstep}
\end{eqnarray}
Note that
$\sum_{p=1}^h\vol_p\leq\V^x(U)+\beta(h-1)$ and $\sum_{p=1}^h|\T_p|\leq N-h$ (since each
time we create a child of $(U,\T)$ we remove at least one new $(s_i,t_i)$ pair from
$\T$). 
So adding \eqref{indstep} over all $p=1,\ldots,h$, we obtain that 
$c(Z_U)\leq 
4\bigl(\beta|\T|+\V^x(U)\bigr)\ln\bigl(\frac{e^dr\V^x(U)}{\OPT}\bigr)\ln\ln\bigl(e(r+1)\bigr)$.
\end{proof}

\begin{theorem} 
\label{2mcthm}
Algorithm \twomcalg returns a feasible solution of cost at most 
$O(\ln r\ln\ln r)\cdot\OPT$. 
\end{theorem}

\begin{proof}
Feasibility of $Z$ is shown in Lemma~\ref{2feas}.
Each time we make a recursive call to \twomcalg, the number of source-sink pairs involved
decreases by at least a factor of 2, so the depth $d$ of the overall recursion tree is
$O(\log_2 r)$. Since $\beta r+\V^x(V)=\bigl(2+\frac{1}{r}\bigr)\OPT$ and
$\frac{r\V^x(V)}{\OPT}=r+1$, by Lemma~\ref{indnlem}, this implies that 
$c(Z)\leq O(\OPT)\cdot\bigl(\ln(r+1)+O(\log_2 r)\bigr)\ln\ln\bigl(e(r+1)\bigr)$.
\end{proof}

\begin{lemma} \label{2feas}
The solution $Z$ returned by \twomcalg is feasible.
\end{lemma}

\begin{proof}
Suppose for a contradiction that there is some $s_i$-$t_i$ pair such that there are (at
least) 2 edge-disjoint $s_i$-$t_i$ paths $P_1, P_2$ in $(V,E\sm Z)$. 
Consider the recursion tree of \twomcalg, and let $(Y,\T_Y)$ be the node furthest from the
root such that $P_1,P_2\sse E(Y)$.  
(Such a node must exist since the root satisfies this property.)

Note that there is at least one child $(X,.)$ of $(Y,\T_Y)$ such that 
$\dt^Y_{P_1\cup P_2}(X)\neq\es$. If not and both $P_1$ and $P_2$ are contained in 
some child of $(Y,\T_Y)$ then this contradicts the definition of $(Y,\T_Y)$.  
Otherwise, we have $P_1,P_2\sse E(A)$, where 
$A=Y\sm\bigcup_{\text{children $(X,.)$ of $(Y,\T_Y)$}}X$. But then we would have processed
$A$ in step A3 and created at least one child $(A',.)$ for some $A'\sse A$. 

We claim that if $\dt^Y_{P_1\cup P_2}(X)\neq\es$ for a child $(X,.)$ of $(Y,\T_Y)$, then
$|\dt^Y_{P_1\cup P_2}(X)|\geq 2$. This is true if both $s_i$ and $t_i$ are in $X$ or if
neither of them are in $X$ since then a path crossing $X$ must cross it at least
twice. Otherwise, $X$ is an $s_i$-$t_i$ cut, and since $P_1$ and $P_2$ are edge-disjoint
$s_i$-$t_i$ paths in $E(Y)$, we again have $|\dt^Y_{P_1\cup P_2}(X)|\geq 2$.
Among all the children $(X,.)$ of $(Y,\T_Y)$ such that $\dt^Y_{P_1\cup P_2}(X)\neq\es$, let
$(W,.)$ be the child that was added to $\Sc$ earliest in step A3.4 of $\twomcalg(Y,\T_Y)$;
let $S'\sse Y$ be the current set $S$ when $W$ was added. Then, $P_1\cup P_2\sse E(S')$,
and so $|\dt^{S'}_{E\sm Z}(W)|\geq |\dt^{S'}_{P_1\cup P_2}(W')|\geq 2$. But this is a
contradiction, since we include in $Z$ all but at most one edge of $\dt^{S'}(W)$.
\end{proof}

\vspace{-1ex}
\paragraph{\boldmath General $k$ and unit costs.}
The algorithm, which we denote by \kmcalg, leading to our bicriteria guarantee is quite
similar to \twomcalg.  
The only changes are the following:

\begin{list}{$\bullet$}{\itemsep=0ex \addtolength{\leftmargin}{-2ex}}
\item In steps A1 and A3.5, we set $\T'$ to be the $s_i$-$t_i$ pairs from $\T$ that are 
at least $\gm(k-1)$-edge-connected in $(S,E(S))$. 
\item In step A3.2, we apply Corollary~\ref{cor3} with the set $S$ to find the radius
$\rho_1\in[0,1)$.
\item In step A3.4, we add {\em all} edges of $\dt^S(A)$ to $Z$. 
(Unlike 2-\mc, if we only include the edges contributing to $c^q\bigl(\dt^S(A)\bigr)$
for some suitable $q$, then we cannot necessarily argue that the final solution satisfies
the stated connectivity guarantee.) 
\item Of course, in step A4, we now recursively call \kmcalg (with the same arguments).
\end{list}

\begin{theorem} 
\label{kmcthm}
For any $\gm>1$, algorithm \kmcalg returns a solution $Z$ such that 
$c(Z)\leq O\bigl(\frac{\gm}{(\sqrt{\gm}-1)^2}\ln r\ln\ln r\bigr)\cdot\OPT$ and every
$s_i$-$t_i$ pair is less than $\gm(k-1)$-edge-connected in $(V,E\sm Z)$.

Thus, taking $\gm=\frac{k}{k-1}$, we obtain a feasible solution of cost at most
$O\bigl((k-1)^2\ln r\ln\ln r\bigr)\cdot\OPT$.
\end{theorem} 

\begin{proof} 
Let $Z$ be the output of
$\kmcalg\bigl(V,\{(s_1,t_1),\ldots,(s_r,t_r)\}\bigr)$. It is clear that $Z$ is feasible:
every $s_i$-$t_i$ pair that is at least $\gm(k-1)$-edge-connected in $(U,E(U))$, where
$(U,.)$ is a node of the recursion tree is either taken care of (i.e., rendered less
than $\gm(k-1)$-edge-connected) by the edges added to $Z$ in step A3, or, by induction, is
taken care of by a recursive call.

Mimicking the proof of Lemma~\ref{indnlem}, and using Corollary~\ref{cor3} in place of 
Lemma~\ref{cor2} in the proof, one can easily show that if $d$ is the depth of the
recursion tree rooted $(U,\T)$ and $Z_U=\kmcalg(U,\T)$, then 
$$ 
c(Z_U)\leq
\frac{2\gm}{(\sqrt{\gm}-1)^2}\Bigl(\beta|\T|+\V^x(U)\Bigr)%
\ln\Bigl(\frac{e^dr\V^x(U)}{\OPT}\Bigr)\ln\ln\bigl(e(r+1)\bigr).
$$
Since the depth of the overall recursion tree is $O(\log_2 r)$, as argued in the proof of
Theorem~\ref{2mcthm}, we obtain that 
$c(Z)\leq O\bigl(\frac{\gm}{(\sqrt{\gm}-1)^2}\ln r\ln\ln r\bigr)\cdot\OPT$. 
\end{proof}

\section{\boldmath Improved hardness result for \kmic} \label{kstcut}
Theorems~\ref{thm:inapproxSSVE} and~\ref{kmichard} together prove that \kmic is at least
as hard as the densest-$k$-subgraph problem (\dks) problem. 
In \dks on hypergraphs, we seek a set of $k$ nodes containing the maximum the number of
hyperedges. Our hardness result implies that obtaining a
unicriterion $O\bigl(k^{\e_0}\polylog(n)\bigr)$-approximation for some constant $\e_0$ would
improve the current-best guarantee for \dks on graphs, and imply the existence of a
certain family of one-way functions.
Our reduction is from \emph{small set vertex expansion} (\ssve), wherein we have a
bipartite graph $G=(U \cup V,E)$ and a parameter $0 < \alpha \leq 1$, and we seek a subset
$S \subseteq U$ with $|S| \geq \alpha |U|$ that minimizes the number of neighbors,
$\Gamma(S)$. 
Chuzoy et al.~\cite{ChuzhoyMVZ12} show that \ssve reduces to the minimization version of
\dks, \mindks, wherein we seek a minimum number of nodes that contain at least $k$
hyperedges. They also show that a $\rho$-approximation for \mindks on $\ld$-uniform
hypergraphs yields a $(2\rho^\ld)$-approximation for \dks on $\ld$-uniform hypergraphs. 

For a graph $H=(V_H,E_H)$, subset $S\sse V_H$, and $v\in V_H\sm S$, we use
$\dt_H(S,v)=\dt_H(v,S)$ to denote the edges between $v$ and $S$ in $H$, and 
$\Gm_H(S)$ to denote the set of neighbors of $S$ in $H$. As is standard, we
abbreviate $\dt_H(\{v\},V_H\sm\{v\})$ to $\dt_H(v)$. 

\begin{theorem}[\cite{ChuzhoyMVZ12}] \label{thm:inapproxSSVE}
For any $\lambda \geq 2$, there is a polytime approximation-preserving reduction that
given a \mindks-instance on a $\ld$-uniform hypergraph with $n$ nodes and $m$ edges,
creates an \ssve-instance with $m+n$ nodes and $\ld m$ edges. 
Hence, a $\rho(m,n)$-approximation for \ssve yields, for $\lambda$-uniform hypergraphs, a
$\rho(\ld m,m+n)$-approximation for \mindks and a 
$\bigl(2\bigl(\rho(\ld m,m+n)\bigr)^\lambda\bigr)$-approximation for \dks.
\end{theorem}

\begin{theorem} \label{kmichard}
There is a polytime approximation-preserving reduction that given an \ssve-instance with
$n$ nodes and $m$ edges,
creates a \kmic-instance with $O(n^3)$ nodes, $O(mn^2+n^5)$ edges, and \mbox{$k=O(mn^2)$.}  
Hence, a $\rho(k,m,n)$-approximation for \kmic yields a
$\rho\bigl(O(mn^2),O(mn^2+n^5),O(n^3)\bigr)$-approximation for \ssve. 
\end{theorem}

\begin{proof}
Given an instance $\bigl(G=(U\cup V,E),\al\bigr)$ of \ssve, we construct the following
instance of \kmic; see Fig.~\ref{fig:SSVE}. 
Let $N=2|U||V|+1$. 
Below, an infinite-cost edge $(u,v)$ of capacity $b_{uv}$ is simply a shorthand for
$b_{uv}$ parallel infinite-cost edges. Also, unless otherwise specified, an edge has unit
capacity.  

\begin{list}{(\roman{enumi})}{\usecounter{enumi} \topsep=0.5ex \itemsep=0ex
    \addtolength{\leftmargin}{-1ex}} 
\item We replace each vertex $v \in V$ with a clique
$K(v)$ of size $N$. All edges in the clique have infinite cost. 
For each edge $(u,v) \in E$, we add an edge between $u$ and every vertex in $K(v)$.  

\item We add the source $s$ and connect it to all vertices in $U$; 
we add the sink $t$ and connect it to all vertices in $K(v)$ for every $v\in V$.

\item Finally, we add a vertex $b$, an edge $(b,t)$ with capacity $|E| \cdot N$, and edges
$(b,u)$ with capacity $|\delta_G(u)| \cdot N$ for all $u\in U$. 
We also add a vertex $a$, an edge $(s,a)$ with capacity $|E| \cdot N$, and for all 
$v \in V$, we add edges $(a,x)$ for all $x\in K(v)$ with capacity $|\dt_G(v)|$.
\end{list}

\noindent
All edges have infinite cost except for the edges between $\bigcup_{v\in V}K(v)$ and $t$,
which have unit cost. 
We set $k = |U|(1- \alpha) + N|E| +1$. 

\begin{figure}[t!]
\centering
\includegraphics[width = 0.6\textwidth]{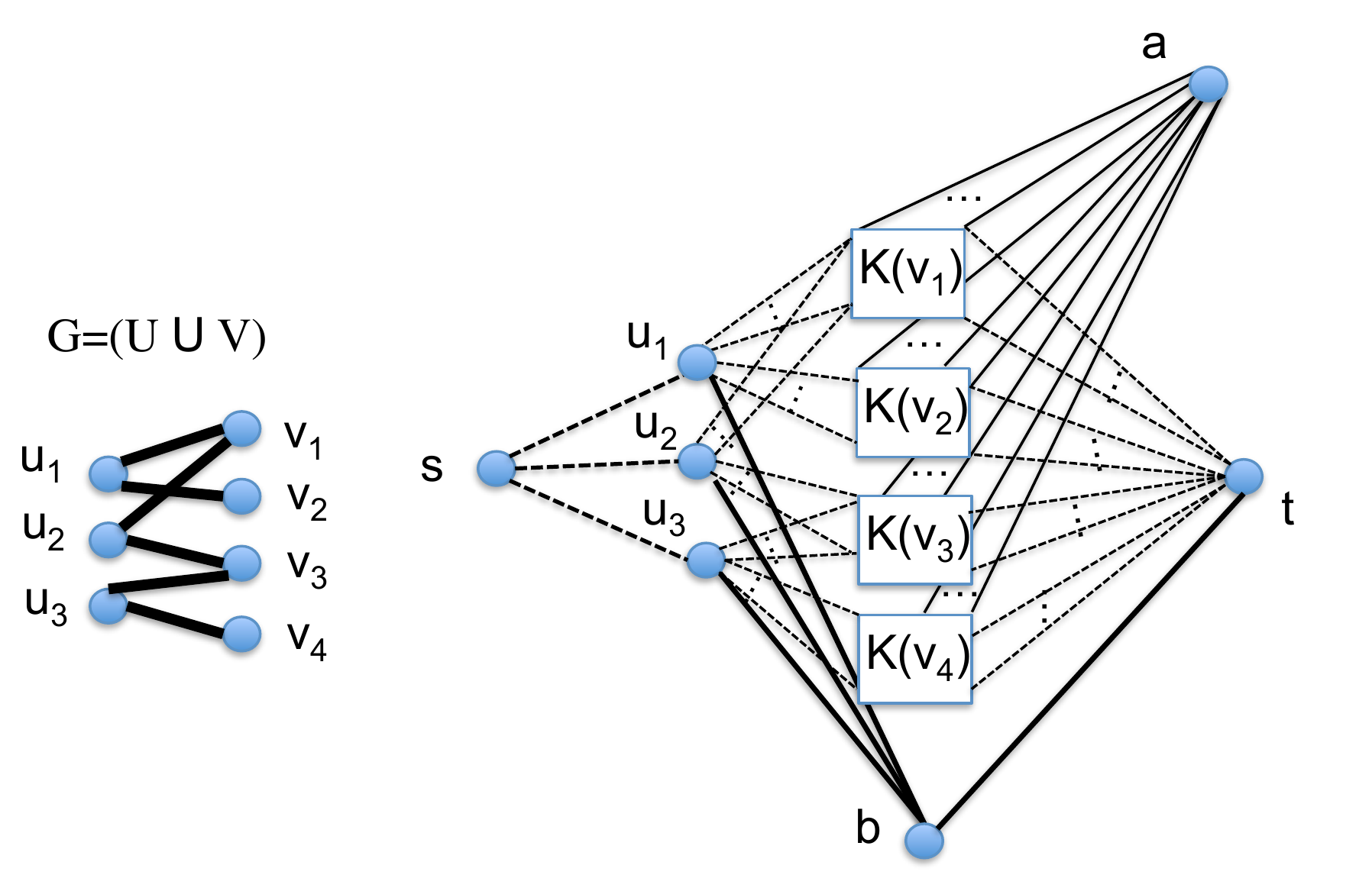} 
\caption{To the left, a graph of a given SSVE instance. To the right, the graph for the  \kmic instance created by our reduction.
The edges incident into $t$ have unit cost, while all other edges have infinite cost. 
Each $K(v_i)$ is a clique with $N=25$ vertices. Dashed edges have unit capacity. The other edges have the 
following capacities: edge $\{s,a\}$ and edge $\{b,t\}$ have capacity 150. Each edge $\{b,u_i\}$ ($i=1,2,3$) has 
capacity 50. Each edge $\{a,x\}$ for $x \in K(v_1)$ and $x \in K(v_3)$ has capacity 2. Each edge $\{a,x\}$ for $x \in K(v_2)$ and $x \in K(v_4)$ has capacity 1. }
\label{fig:SSVE}
\end{figure}

We claim that there exists a solution of value at most $C$ for the \ssve instance
iff there is a solution of value at most $CN$ for the \kmic instance.
Note that a solution $F$ consisting of unit-cost edges is feasible 
if the maximum $s$-$t$ flow in the capacitated remainder graph $G'\sm
F$ has value at most $k-1$. 
The intuition is that if we send $N|E|$ units of flow along the paths $s-a-x-u-b-t$ for
all $(u,v)\in E, x\in K(v)$, then in the residual digraph, all arcs between $U$ and
$\bigcup_{v\in V}K(v)$ leave $U$. Given this, one can mimic the arguments
in~\cite{ChuzhoyMVZ12} to show the desired claim.

Suppose the \ssve instance has a solution $S\sse U$ of value at most $C$. 
Construct a \kmic-solution by removing the $(v,t)$ edges for all $v\in\Gamma_G(S)$. 
Clearly the cost of this set is at most $CN$. We now argue feasibility. 
Consider the $s$-$t$ cut induced by $A=\{s,a\}\cup S\cup\bigcup_{v\in\Gamma_G(S)}K(v)$.
Then $|\dt_{G'}(A)|$ is equal to
$$
\underbrace{|U\sm S|}_{\normalsize{\substack{\text{edges between}\\ \text{$s$ and $U\sm S$}}}}+\
\underbrace{\sum_{u\in S}N|\dt_G(u)|}_{\substack{\text{edges between}\\ \text{$b$ and
      $S$}}}\ +
\underbrace{\sum_{v\in V\sm\Gm_G(S)}N|\dt_G(v)|}_{\substack{\text{edges
    between $a$ and}\\ \text{$\bigcup_{v\in V\sm\Gm_G(s)}K(v)$}}}\ +\ 
\underbrace{N\bigl(\text{\# of edges in $G$ between $U\sm S$ and
    $\Gm_G(S)$}\bigr)}_{\text{edges between $U\sm S$ and  
    $\bigcup_{v\in\Gm_G(S)}K(v)$}}.
$$
The sum of the last two terms is $\sum_{u\in U\sm S}N|\dt_G(u)|$ and 
$|U\sm S|\leq|U|(1-\al)$, so the size of the $s$-$t$ cut is at most
$|U|(1-\al)+\sum_{u\in U}N|\dt_G(u)|\leq|U|(1-\al)+N|E|\leq k-1$.

For the other direction, suppose $G'$ has a solution $F$ of value at most $CN$. 
Clearly, $F$ can consist of only unit-cost edges (incident to $t$).
We first argue that we may convert $F$ into a structured feasible solution $F'$ of cost at
most $CN$ where $|F'\cap\dt_{G'}(K(v))|\in\{0,N\}$ for all $v\in V$. 

Fix $v\in V$.
If $|\delta_{G'}(K(v),t)\sm F|\leq k':=|U|(1- \alpha)$, then we add all edges of
$\dt_{G'}(K(v),t)\sm F$  
to $F$. 
Now suppose $|\delta_{G'}(K(v),t)\sm F|>k'$ and let 
$w_1,\dots,w_{k'+1}$ be vertices in $K(v)$ such that $(w_i,t)\notin F$ for all
$i=1,\ldots,k'+1$. We claim that $F\sm\dt_{G'}(K(v),t)$ 
is also feasible. 
Suppose to the contrary that we now have $k'+1$ edge-disjoint $s$-$t$-paths. 
We may assume that each such path contains at most one vertex from $K(v)$ since we can
always shortcut the path to $t$. Consider a path $P$ that contains a vertex $w\in K(v)$
where $w\notin\{w_1, \dots, w_{k'+1}\}$.    
Then we can construct another path $P'$ by switching $w$ with a distinct vertex $w_j$ for
some $j\in\{1\dots, k'+1\}$. Note that $P'$ is an $s$-$t$ path that avoids edges in $F$.
If we repeat this argument for all such paths, we obtain $k'+1$ edge-disjoint
$s$-$t$-paths in $G'\setminus F$, contradicting the feasibility of $F$.

If we perform the above transformation for all $v\in V$, then we obtain a feasible
solution $F'$ of cost at most $|F|+|V|k'<(C+1)N$.  
But by construction $|F'|$ must be an integer multiple of $N$, so $|F'|\leq CN$.

Consider the residual network $\tilde G$ that is obtained from $G'\sm F'$ as follows. We 
first bidirect the edges of $G'\sm F'$, giving each resulting arc the same capacity as that of
the corresponding edge of $G'$. $\tG$ is the residual network obtained after we send one
unit of flow along the path $s$-$a$-$x$-$u$-$b$-$t$ for every edge $(u,v)\in E$ and every 
$x\in K(v)$. Note that these paths are edge disjoint, so we send $N|E|$ units of flow.
By flow theory~\cite{AhujaMO93}, we know that the value of the maximum $s$-$t$-flow in
$G'\sm F'$ is at most $k':=|U|(1-\al)$ iff the value of maximum $s$-$t$-flow in $\tG$,
which equals the capacity of the minimum $s$-$t$ cut in $\tG$, is at most $k'$. 
It follows there is an $s$-$t$ cut in $\tG$ of capacity at most $k'$.

Let $A$ be the vertices that are on the $s$-side of this cut. 
Let $S := U \cap A$. Then $|S|\geq\alpha|U|$, otherwise the cut would have capacity more
than $|U|(1- \alpha)$ due to the arcs $(s,u)$ for $u\in U\setminus S$.
Consider $v\in\Gm_G(S)$. We must have $K(v)\sse A$: if $K(v)\cap A=\es$, then considering
$u\in S$ such that $(u,v)\in E$, the cut has capacity at least $N>k'$ due to the edges
between $u$ and $K(v)$; otherwise, since $K(v)$ is split between the $s$- and $t$- sides,
the cut has capacity at least $N-1>k'$.
Finally, $\dt_{G'}(K(v),t)\sse F'$, otherwise $\dt_{G'}(K(v),t)\cap F'=\es$, and
again the cut has capacity at least $N>k'$. Thus, $|\Gm_G(S)|\leq C$, so $S$ is an
\ssve-solution. 
\end{proof}

\section{\boldmath Extensions to node-connectivity versions of \kmc} \label{node-extn}
We now consider variants of \kmc where we seek to delete edges or nodes so as to reduce
the {\em node connectivity} of each $s_i$-$t_i$ pair to at most $k-1$. 
Formally, as before, we are given an undirected graph $G=(V,E)$, $r$ source-sink pairs
$(s_1,t_1),\ldots,(s_r,t_r)$, and an integer $k\geq 1$.
In the {\em edge-deletion $k$-route node-multicut} (\edknmc) problem, we have nonnnegative
edge-costs $\{c_e\}_{e\in E}$ and we seek a minimum-cost set $F\sse E$ of edges such that 
the remainder graph $\bG=(V,E\sm F)$ contains at most $k-1$ node-disjoint
$s_i$-$t_i$ paths for all $i=1,\ldots,r$. 
In the {\em node-deletion $k$-route node-multicut} (\ndknmc) problem, we have nonnegative
node costs $\{c_v\}_{v\in V}$ and we seek a minimum-cost set 
$A\sse V\sm\{s_1,t_1,\ldots,s_r,t_r\}$ of nodes such that 
the remainder graph $\bG=\bigl(V\sm A,E(V\sm A)\bigr)$ 
contains at most $k-1$ node-disjoint $s_i$-$t_i$ paths for all $i=1,\ldots,r$.

The LP-relaxations of these $k$-route node-multicut problems induce both edge and node
lengths, so to round these we develop region-growing lemmas that also incorporate
node lengths. To keep notation simple, instead of proving an overly- general
region-growing lemma and obtaining the lemmas required for \edknmc and \ndknmc as
corollaries, we specifically focus on \edknmc (Section~\ref{edknmc}) and \ndknmc
(Section~\ref{ndknmc}) and prove suitable region-growing lemas. 
We use these to obtain an $O(\ln r\ln\ln r)$-approximation for \edknmc with $k=2$, and a
bicriteria 
$\bigl(\gm,O\bigl(\frac{\gm}{(\sqrt{\gm}-1)^2}\ln r\ln\ln r\bigr)\bigr)$-approximation  
for \ndknmc with general $k$ and unit node costs.

\subsection{\boldmath Edge-deletion $k$-route node-multicut (\edknmc) with $k=2$} \label{edknmc}
The LP-relaxation for \edknmc is as follows.
\begin{alignat*}{3}
\min & \quad & \sum_e c_ex_e & \tag{P2} \label{edknmcp} \\
\text{s.t.} && \sum_{e\in E(P)} x_e &+\sum_{v\in V(P)}y^i_v \geq 1 \qquad && 
\forall i, \forall P\in\Pc_i \\
&& \sum_v y^i_v & \leq k-1, \quad y^i_{s_i}=y^i_{t_i}=0 \qquad && \forall i \\
&& x, y & \geq 0. 
\end{alignat*}

\paragraph{Region-growing lemma.}
Let $\bigl(x,\{y^i\}\bigr)$ be an optimal solution to \eqref{edknmcp}, and $\OPT$ be its
value.  
Let $S\sse V$ represent the node-set of the current region. 
For $T\sse S\sse V$, recall that $E(S)$ is the set of edges with both endpoints in $S$
and $\dt^S(T)$ is the set of boundary edges of $T$ in $E(S)$.
Set $\beta=\OPT/r$. 
As before, define $\V^x(S):=\beta+\sum_{e\in E(S)}c_ex_e$. 
Let $\rho\geq 0$. Let $z\in V$. Fix a commodity $i$.

\begin{list}{$\bullet$}{\itemsep=0ex \topsep=0.5ex \addtolength{\leftmargin}{-2ex}} 
\item Define ${\ell^i}(u;v)=\min_{P:P\text{ is a $u$-$v$ path}}
\bigl(\sum_{e\in E(P)}x_e+\sum_{w\in V(P):w\neq u}y^i_w\bigr)$, where $E(P)$ and
$V(P)$ denote respectively the set of edges and nodes of $P$.
Note that 
${\ell^i}$ defines an {\em asymmetric metric} on $V\times V$.

\item Define $B_i(z,\rho):=\{v\in V: {\ell^i}(z;v)\leq\rho\}$ to be the ball of radius
$\rho$ around $z$. Let $B^S_i(z,\rho):=B_i(z,\rho)\cap S$. 

\item Define the edge-boundary $\del^{S,x}_i$, and node-boundary $\Gm^{S,y}_i$, of
$B^S_i(z,\rho)$ in $S$ as follows.
\begin{eqnarray*}
\del^{S,x}_i(z,\rho) & := & \{(u,v)\in E: u,v\in S,\ {\ell^i}(z;u)\leq\rho,\ {\ell^i}(z;v)-y^i_v>\rho\} \\
\Gm^{S,y}_i(z,\rho) & := & \{v\in S: \rho<{\ell^i}(z;v)\leq\rho+y^i_v\}.
\end{eqnarray*}
Let $\ov{B^S_i}(z,\rho):=S\sm\bigl(B^S_i(z,\rho)\cup\Gm^{S,y}_i(z,\rho)\bigr)$.

\item Define the following volumes:
\begin{eqnarray*}
\V^{S,x}_i(z,\rho) & := & \beta
+\sum_{\substack{(u,v)\in E: u\in B^S_i(z,\rho) \\ v\in B^S_i(z,\rho)\cup\Gm^{S,y}_i(z,\rho)}}c_{uv}x_{uv}
+\sum_{\substack{(u,v)\in\del^{S,x}_i(z,\rho): \\ u\in B^S_i(z,\rho)}}c_{uv}\bigl(\rho-\ell^i(z;u)\bigr) \\[0.5ex]
\ov{\V^{S,x}_i}(z,\rho) & := & \beta
+\sum_{\substack{(u,v)\in E: u\in\ov{B^S_i}(z,\rho) \\ v\in\ov{B^S_i}(z,\rho)\cup\Gm^{S,y}_i(z,\rho)}}c_{uv}x_{uv}
+\sum_{\substack{(u,v)\in\del^{S,x}_i(z,\rho): \\ u\in B^S_i(z,\rho)}}c_{uv}\bigl(\ell^i(z;v)-\rho-y^i_v\bigr)
\end{eqnarray*}
\end{list}

\begin{lemma} \label{ednreggrow}
Let $S\sse V$, $z\in V$, $i$ be some commodity, and $0\leq a<b$.
Let $\rho$ be chosen uniformly at random from $[a,b)$. Then,
\begin{eqnarray}
\E[\rho]{\frac{c\bigl(\del^{S,x}_i(z,\rho)\bigr)}
{\V^{S,x}_i(z,\rho)\ln\Bigl(\frac{e\V^{S,x}_i(z,b)}{\V^{S,x}_i(z,\rho)}\Bigr)}}
& \leq & \frac{1}{b-a}\cdot\ln\ln\biggl(\frac{e\V^{S,x}_i(z,b)}{\V^{S,x}_i(z,a)}\biggr) 
\quad\text{and} \label{ednbnd1} \\
\E[\rho]{\frac{c\bigl(\del^{S,x}_i(z,\rho)\bigr)}
{\ov{\V^{S,x}_i}(z,\rho)\ln\Bigl(\frac{e\ov{\V^{S,x}_i}(z,a)}{\ov{\V^{S,x}_i}(z,\rho)}\Bigr)}}
& \leq &
\frac{1}{b-a}\cdot\ln\ln\biggl(\frac{e\ov{\V^{S,x}_i}(z,a)}{\ov{\V^{S,x}_i}(z,b)}\biggr) 
\label{ednbnd2}
\end{eqnarray}
\end{lemma}

\begin{proof} 
We abbreviate $c\bigl(\del^{S,x}_i(z,\rho)\bigr)$ to $c(\rho)$,
$\V^{S,x}_i(z,\rho)$ to $\V(\rho)$ and $\ov{\V^{S,x}_i}(z,\rho)$ to
$\ov{\V}(\rho)$.
Let $I=\{\ell^i(z;v), \ell^i(z;v)-y^i_v: v\in V\}$. Note that $\V(\rho)$ and
$\ov{\V}(\rho)$ are differentiable at all $\rho\in[a,b)\sm I$ and for
each such $\rho$, we have
$\frac{d\V(\rho)}{d\rho}=c(\rho)$ and 
$\frac{d\ov{\V}(\rho)}{d\rho}=-c(\rho)$. 
The proof now follows from exactly the same arguments as in the proof of
Lemma~\ref{reggrow}. 
\end{proof}

\begin{corollary} \label{edncor2}
Let $S\sse V$, $z\in V$, and $i$ be some commodity. 
Let $\al\in(0,1)$ and $q=\ceil{\frac{k-1}{\al}}$.
We can efficiently find $\rho_1\in[0,1)$ such that
\begin{eqnarray*}
c\bigl(\del^{S,x}_i(z,\rho_1)\bigr) & \leq &
\frac{2}{1-\al}\cdot\V^{S,x}_i(z,\rho_1)\ln\Bigl(\frac{e\V^{x}(S)}{\V^{S,x}_i(z,\rho_1)}\Bigr)%
\ln\ln\bigl(e(r+1)\bigr) \\ 
c\bigl(\del^{S,x}_i(z,\rho_1)\bigr) & \leq &
\frac{2}{1-\al}\cdot\ov{\V^{S,x}_i}(z,\rho_1)\ln\Bigl(\frac{e\V^{x}(S)}{\ov{\V^{S,x}_i}(z,\rho_1)}\Bigr)%
\ln\ln\bigl(e(r+1)\bigr) \\ 
\bigl|\Gm^{S,y}_i(z,\rho_1)\bigr| & < & q. 
\end{eqnarray*}
\end{corollary}

\begin{proofnobox}
If we choose $\rho$ uniformly at random from $[0,1)$ then
$\E[\rho]{\bigl|\Gm^{S,y}_i(z,\rho)\bigr|}\leq\sum_iy^i_v\leq k-1$.
Taking $[a,b)=[0,1)$, the arguments in Corollary~\ref{cor1} readily generalize to show
that we can efficiently find $\rho_1\in[0,1)$ such that 
${c\bigl(\del^{S,x}_i(z,\rho_1)\bigr)}/
{\V^{S,x}_i(z,\rho_1)\ln\bigl(\frac{e\V^{S,x}_i(z,1)}{\V^{S,x}_i(z,\rho_1)}\bigr)}$, and 
${c\bigl(\del^{S,x}_i(z,\rho_1)\bigr)}/
{\ov{\V^{S,x}_i}(z,\rho_1)\ln\bigl(\frac{e\ov{\V^{S,x}_i}(z,0)}{\ov{\V^{S,x}_i}(z,\rho_1)}\bigr)}$
are at most $\frac{2}{(1-\al)}$ times the right-hand-sides of \eqref{ednbnd1} and
\eqref{ednbnd2} respectively, and $\bigl|\Gm^{S,y}_i(z,\rho)\bigr|<\frac{k-1}{\al}$. 
The lemma now follows by noting that 
\vspace{-1ex}
\begin{equation*}
\V^{S,x}_i(z,1)\leq\V^x(S), \qquad \ov{\V^{S,x}_i}(z,0)\leq\V^x(S), \qquad
\frac{\V^{S,x}_i(z,1)}{\V^{S,x}_i(z,0)}\leq r+1, \qquad
\frac{\ov{\V^{S,x}_i}(z,0)}{\ov{\V^{S,x}_i}(z,1)}\leq r+1. 
\tag*{\qedsymbol}
\end{equation*}
\end{proofnobox}

\vspace{-1ex}
\paragraph{Algorithm and analysis.} 
The algorithm and analysis dovetail the one in Section~\ref{kmcalg} for $k=2$.

\vspace{1.5ex}

{\small \hrule \vspace{5pt}
\noindent {\bf Algorithm} $\edtwonmcalg(U, \T=\{(s_1,t_1),\ldots,(s_{N},t_{N})\}$) 

\noindent {\bf Input:} A subset $U\sse V$, and a collection $\T=\{(s_i,t_i)\}_{i=1}^{N}$ of
source-sink pairs, where $s_i,t_i\in U$ for all $i=1,\ldots, N$.

\noindent {\bf Output:} A set $Z\sse E(U)$ such that $s_i$ and $t_i$ are at most
1-node-connected in $(U,E(U)\sm Z)$ for all $i=1,\ldots,N$.

\begin{list}{B\arabic{enumi}.}{\usecounter{enumi} \topsep=0.5ex \itemsep=0.25ex \addtolength{\leftmargin}{-1.5ex}}
\item Set $S=U$, $Z=\es$, $\Sc=\es$, and
$\T'=\bigl\{(s_i,t_i)\in\T: \text{$s_i$ and $t_i$ are at least $2$-node-connected in $(S,E(S))$}\bigr\}$. 

\item If $\T'=\es$, return $Z$.

\item While $\T'\neq\es$, we do the following.
\begin{list}{A\arabic{enumi}.\arabic{enumii}}{\topsep=0.25ex \itemsep=0ex \usecounter{enumii}}
\item Pick some $(s_i,t_i)\in\T'$.
\item Apply Corollary~\ref{edncor2} with $z=s_i$, $\al=0.5$ and the set $S$ (and $k=2$) to 
find a radius $0\leq\rho_1<1$. 
\item If $B^S_i(s_i,\rho_1)\cup\Gm^{S,y}_i(s_i,\rho_1)$ contains at most
$N/2$ pairs from $\T$ then set $A=B^S_i(s_i,\rho_1)$, else set $A=\ov{B^S_i}(s_i,\rho_1)$. 
\item Set $\Sc\assign\Sc\cup\{A\cup\Gm^{S,y}_i(s_i,\rho_1)\}$. 
Add the edges in $\del^{S,x}_i(s_i,\rho_1)$ to $Z$. 
\item  Set $S\assign S\sm A$. 
Update $\T'$ to be the $s_i$-$t_i$ pairs from $\T$ that are at least $2$-node-connected in
$(S,E(S))$.
\end{list}

\item For every set $A\in\Sc$, set 
$Z\assign Z\cup\edtwonmcalg\bigl(A,\{(s_i,t_i)\in\T: s_i,t_i\in A\}\bigr)$.

\item Return $Z$.
\end{list}

\noindent The initial call to $\edtwonmcalg$, which computes the solution we return, is
$\edtwonmcalg\bigl(V,\{(s_1,t_1),\ldots,(s_r,t_r)\}\bigr)$. 
\hrule
}

\vspace{1.5ex}
Let $Z:=\edtwonmcalg\bigl(V,\{(s_1,t_1),\ldots,(s_r,t_r)\}\bigr)$.
Define the depth of a subtree of the recursion tree corresponding to the execution of
\edtwonmcalg to be the maximum number of edges on a root to leaf path of the subtree.
The following claim will be useful to prove feasibility of $Z$.

\begin{claim} \label{nodecon}
Let $S, T\sse V$ with $|S\cap T|\leq 1$. 
Let $E_S\sse E(S)$ and $E_T\sse E(T)$. Let $u,v\in S$ be such that $u$ and $v$ are at most
1-node-connected in $(S,E_S)$. Then, $u$ and $v$ are at most 1-node-connected in 
$(S\cup T,E_S\cup E_T)$.
\end{claim} 

\begin{proof}
If $S\cap T=\es$, this clearly holds. So assume otherwise.
Suppose $P_1$, $P_2$ are two simple node-disjoint $u$-$v$ paths in 
$G'=(S\cup T,E_S\cup E_T)$. At least one of $P_1$ and $P_2$ does not lie completely in
$(S,E_S)$; suppose $P_2$ is this path. But since all edges of $\dt_{G'}(S)$ are incident 
to a single node, and $P_2$ both exits and leaves $S$, $P_2$ contains a repeated node,
which is a contradiction.
\end{proof}

\begin{lemma} \label{edn2feas}
The solution $Z$ returned is feasible.
\end{lemma}

\begin{proof}
Suppose for a contradiction that there is some $s_i$-$t_i$ pair that is (at least)
2-node-connected in $(V,E\sm Z)$. 
Consider the recursion tree of \edtwonmcalg, and let $(Y,\T_Y)$ be the node furthest from the
root such that $s_i$ and $t_i$ are at least 2-node-connected in the subgraph $(Y,E(Y))$
induced by $Y$. 
Suppose that the loop in step B3 of $\edtwonmcalg(Y,\T_Y)$ runs for $h$ iterations. 
Note that $h\geq 1$ since $s_i$ and $t_i$ are at least 2-node-connected in $(Y,E(Y))$.
Let $X_p$ be the set added to $\Sc$ in step B3.4 in the $p$-th iteration of the loop. 
Let $X_{h+1}\sse Y$ be the set $S$ at the termination of the loop.
Let $S_p=\bigcup_{q=p}^{h+1}X_q$ (so $S_1=Y$).
Notice that $|X_p\cap S_{p+1}|\leq 1$, since
$X_p\cap S_{p+1}\sse\Gm^{S_p,y}_p(s_p,\rho_p)$, where $s_p$-$t_p$ is the
source-sink pair and $\rho_p$ is the radius chosen in the $p$-th iteration, and
$|\Gm^{S_p,y}_p(s_p,\rho_p)|<2$ by Lemma~\ref{edncor2}. 

Let $p$ be the highest index such that $s_i$ and $t_i$ are at least 2-node-connected in
$(S_p,E(S_p))$. Note that $p\leq h$, otherwise the loop in step B3 would not
have terminated with $S=X_{h+1}$. If $s_i, t_i\in S_{p+1}$, they are at most
1-node-connected in $(S_{p+1},E(S_{p+1}))$. Since $|X_p\cap S_{p+1}|\leq 1$, we have 
$E(S_p)=E(X_p)\cup E(S_{p+1})$, and by Claim~\ref{nodecon} it follows that $s_i$ and $t_i$ 
are at most 1-node-connected in $(S_p,E(S_p))$, which is a contradiction.
If $s_i, t_i\in X_p$, they are at most 1-node-connected in $(X_p,E(X_p))$ due to the
definition of $(Y,\T_Y)$, and so we arrive at the same contradiction.
So it must be that $\bigl|\{s_i,t_i\}\cap (X_p\sm S_{p+1})\bigr|=1$. But then all
$s_i$-$t_i$ paths in $(S_p,E(S_p))$ contain the singleton node in $X_p\cap S_{p+1}$.
So $s_i$ and $t_i$ are at most 1-node-connected in $(S_p,E(S_p))$, and we have the same
contradiction. 
\end{proof}

\begin{lemma} \label{ednindnlem}
Let $d$ be the depth of a subtree of the recursion tree rooted at 
$\bigl(U\sse V,\T\sse\{(s_1,t_1),\ldots,(s_r,t_r)\}\bigr)$. 
Let $Z_U=\edtwonmcalg(U,\T)$. 
Then $c(Z_U)\leq 
4\bigl(\beta|\T|+\V^x(U)\bigr)\ln\bigl(\frac{e^dr\V^x(U)}{\OPT}\bigr)\ln\ln\bigl(e(r+1)\bigr)$.
\end{lemma}

\begin{proof}
When $d=0$, we have $Z_U=\es$, so the statement holds. 
Suppose in step B3 of $\edtwonmcalg(U,\T)$, we add sets $A_1,\ldots,A_h$ to $\Sc$ (where 
$h\geq 1$), in that order. For $p=1,\ldots,h$, let $S_p$ be the current set $S$ when $A_p$
was added to $\Sc$ in step B3.4, and let $E_p$ be the edge-set added to $Z$ in this step. 
Let $Z_1,\ldots,Z_h$ be the edge-sets returned by the recursive calls to
$\edtwonmcalg(A_1,\T_1),\ldots,\edtwonmcalg(A_h,\T_h)$ in step B4.

Let $\vol_p=\V^{S_p,x}_i(s_i,\rho_1)$ if $A_p=B^{S_p}_i(s_i,\rho_1)\cup\Gm^{S_p,y}_i(s_i,\rho_1)$  
and $\vol_p=\ov{\V^{S_p,x}_i}(s_i,\rho_1)$ if $A_p=\ov{B^{S_p}_i}(s_i,\rho_1)\cup\Gm^{S_p,y}_i(s_i,\rho_1)$.
The key thing to note is that we still have $\V^x(A_p)\leq\vol_p\leq\V^x(S_p)\leq\V^x(U)$
and $\sum_{p=1}^h\vol_p\leq\V^x(U)+\beta(h-1)$. The latter follows since an easy induction 
argument shows that $\sum_{p=q}^h\vol_p\leq\V^x(S_q)+\beta(h-q)$ for all $q=1,\ldots,h$.
Given this, the rest of the proof is identical to that of Lemma~\ref{indnlem}.
\end{proof}

Each recursive call to \edtwonmcalg, reduces the number of source-sink pairs involved
by a factor of at least 2, so the depth $d$ of the entire recursion tree is $O(\log_2 r)$. 
So we have shown the following. 

\begin{theorem} \label{edn2mcthm}
Algorithm \edtwonmcalg returns a feasible solution of cost at most 
$O(\ln r\ln\ln r)\cdot\OPT$.  
\end{theorem}

\subsection{\boldmath Node-deletion $k$-route node-multicut (\ndknmc) with unit costs} \label{ndknmc}
The LP-relaxation for \ndknmc is as follows.
\begin{alignat*}{3}
\min & \quad & \sum_v c_vx_v & \tag{P3} \label{ndknmcp} \\
\text{s.t.} && \sum_{v\in V(P)}(x_v &+y^i_v) \geq 1 \qquad && 
\forall i, \forall P\in\Pc_i \\
&& \sum_v y^i_v & \leq k-1, \quad y^i_{s_i}=y^i_{t_i}=0 \qquad && \forall i \\
&& x, y & \geq 0, \qquad x_v=0 \quad && \forall v\in\{s_1,t_1,\ldots,s_r,t_r\}.
\end{alignat*}

\vspace{-1ex}
\paragraph{Region-growing lemma.}
Let $\bigl(x,\{y^i\}\bigr)$ be an optimal solution to \eqref{ndknmcp}, and $\OPT$ be its
value. 
As before, let $S\sse V$ represent the node-set of the current region. 
Set $\beta=\OPT/r$. 
Let $z\in V$ and $\rho\geq 0$. Fix a commodity $i$.

\begin{list}{$\bullet$}{\itemsep=0ex \topsep=0.5ex \addtolength{\leftmargin}{-2ex}} 
\item Define $\ell^i(u;v)=\min_{P:P\text{ is a $u$-$v$ path}}\sum_{w\in V(P):w\neq u}(x_w+y^i_w)$, 
where $V(P)$ is the set of nodes of $P$.
As before, $\ell^i$ defines an asymmetric metric on $V\times V$.

\item Define $B_i(z,\rho):=\{v\in V: \ell^i(z;v)\leq\rho\}$, and
$B^S_i(z,\rho):=B_i(z,\rho)\cap S$.  

\item Define the $x$-boundary of $B^S_i(z,\rho)$ in $S$ to be 
$\Gm^{S,x}_i(z,\rho):=\{v\in S: \rho+y^i_v<\ell^i(z;v)\leq\rho+x_v+y^i_v\}$.
Define the $y$-boundary of $B^S_i(z,\rho)$ in $S$ to be 
$\Gm^{S,y}_i(z,\rho):=\{v\in S: \rho<\ell^i(z;v)\leq\rho+y^i_v\}$.
Note that $\Gm^{S,x}_i(z,\rho)$ and $\Gm^{S,y}_i(z,\rho)$ partition
$\Gm^S_i(z,\rho):=
\{v\in S\sm B^S_i(z,\rho): \exists u\in B^s_i(z,\rho)\text{ s.t. }(u,v)\in E\}$.
Let $\ov{B^S_i}(z,\rho):=S\sm\bigl(B^S_i(z,\rho)\cup\Gm^{S}_i(z,\rho)\bigr)$.

\item Define the following volumes:
\begin{eqnarray*}
\V^{S,x}_i(z,\rho) & := & \beta
+\sum_{u\in B^S_i(z,\rho)\cup\Gm^{S,y}_i(z,\rho)}c_{u}x_{u}
+\sum_{u\in\Gm^{S,x}_i(z,\rho)}c_{u}\bigl(\rho-(\ell^i(z;u)-x_u-y^i_u)\bigr) \\[0.5ex]
\ov{\V^{S,x}_i}(z,\rho) & := & \beta
+\sum_{u\in\ov{B^S_i}(z,\rho)\cup\Gm^{S,y}_i(z,\rho)}c_{u}x_{u}
+\sum_{u\in\Gm^{S,x}_i(z,\rho)}c_u(\ell^i(z;u)-y^i_u-\rho)
\end{eqnarray*}
\end{list}

The following lemma is analogous to Lemma~\ref{ednreggrow} and
follows from the same reasoning.

\begin{lemma} \label{ndnreggrow}
Let $S\sse V$, $z\in V$, $i$ be some commodity, and $0\leq a<b$.
Let $\rho$ be chosen uniformly at random from $[a,b)$. Then,
\begin{eqnarray*}
\E[\rho]{\frac{c\bigl(\Gm^{S,x}_i(z,\rho)\bigr)}
{\V^{S,x}_i(z,\rho)\ln\Bigl(\frac{e\V^{S,x}_i(z,b)}{\V^{S,x}_i(z,\rho)}\Bigr)}}
& \leq & \frac{1}{b-a}\cdot\ln\ln\biggl(\frac{e\V^{S,x}_i(z,b)}{\V^{S,x}_i(z,a)}\biggr) 
\quad\text{and} \\ 
\E[\rho]{\frac{c\bigl(\Gm^{S,x}_i(z,\rho)\bigr)}
{\ov{\V^{S,x}_i}(z,\rho)\ln\Bigl(\frac{e\ov{\V^{S,x}_i}(z,a)}{\ov{\V^{S,x}_i}(z,\rho)}\Bigr)}}
& \leq &
\frac{1}{b-a}\cdot\ln\ln\biggl(\frac{e\ov{\V^{S,x}_i}(z,a)}{\ov{\V^{S,x}_i}(z,b)}\biggr) 
\end{eqnarray*}
\end{lemma}

\begin{corollary} \label{ndncor3}
Let $S\sse V$. Suppose that $s_i,t_i\in S$ and there are $\gm(k-1)$ node-disjoint
$s_i$-$t_i$ paths internal to $S$, where $\gm>1$. Suppose that $c_v=1$ for all nodes $v$.   
We can efficiently find $\rho_1\in[0,1)$ such that
\begin{eqnarray*}
c\bigl(\Gm^S_i(s_i,\rho_1)\bigr)\leq
\frac{2\gm}{(\sqrt{\gm}-1)^2}\cdot\V^{S,x}_i(s_i,\rho_1)\ln\Bigl(\frac{e\V^x(S)}{\V^{S,x}_i(s_i,\rho_1)}\Bigr)%
\ln\ln\bigl(e(r+1)\bigr) \\
c\bigl(\Gm^S_i(s_i,\rho_1)\bigr)\leq
\frac{2\gm}{(\sqrt{\gm}-1)^2}\cdot\ov{\V^{S,x}_i}(s_i,\rho_1)\ln\Bigl(\frac{e\V^x(S)}{\ov{\V^{S,x}_i}(s_i,\rho_1)}\Bigr)%
\ln\ln\bigl(e(r+1)\bigr)
\end{eqnarray*}
\end{corollary}

\begin{proof} 
Let $\al\in(0,1)$, whose value we will fix later. Take $[a,b)=[0,1)$. We have
\vspace{-1ex}
$$
\V^{S,x}_i(z,1)\leq\V^x(S), \qquad \ov{\V^{S,x}_i}(z,0)\leq\V^x(S), \qquad
\frac{\V^{S,x}_i(z,1)}{\V^{S,x}_i(z,0)}\leq r+1, \qquad
\frac{\ov{\V^{S,x}_i}(z,0)}{\ov{\V^{S,x}_i}(z,1)}\leq r+1. 
$$
\vspace{-1ex}
Given this, the arguments in Corollary~\ref{cor1} readily generalize to show 
that we can efficiently find $\rho_1\in[0,1)$ such that 
\begin{eqnarray}
c\bigl(\Gm^{S,x}_i(s_i,\rho_1)\bigr) & \leq & 
\frac{2}{1-\al}\cdot\V^{S,x}_i(z,\rho_1)\ln\Bigl(\frac{e\V^{x}(S)}{\V^{S,x}_i(z,\rho_1)}\Bigr)%
\ln\ln\bigl(e(r+1)\bigr) \label{ndnxbound1} \\
c\bigl(\Gm^{S,x}_i(s_i,\rho_1)\bigr) & \leq &
\frac{2}{1-\al}\cdot\ov{\V^{S,x}_i}(z,\rho_1)\ln\Bigl(\frac{e\V^{x}(S)}{\ov{\V^{S,x}_i}(z,\rho_1)}\Bigr)%
\ln\ln\bigl(e(r+1)\bigr) \label{ndnxbound2} \\
\bigl|\Gm^{S,y}_i(z,\rho)\bigr| & < & \frac{k-1}{\al}. \label{ndnybound}
\end{eqnarray}

Note that $t_i\notin\Gm^S_i(s_i,\rho_1)$ since $\rho_1<1$. So removing
$\Gm^S_i(s_i,\rho_1)$ disconnects $s_i$ and $t_i$, and hence,
$|\Gm^S_i(s_i,\rho_1)|\geq\gm(k-1)$.  
Therefore, since we have unit node costs and $\Gm^{S,x}_i(s_i,\rho_1)$ and
$\Gm^{S,y}_i(s_i,\rho_1)$ partition $\Gm^S_i(s_i,\rho_1)$, we have
$c\bigl(\Gm^S_i(s_i,\rho_1)\bigr)\leq c\bigl(\Gm^{S,x}_i(s_i,\rho_1)\bigr)\cdot\frac{\gm}{\gm-1/\al}$.
Plugging in the bounds from \eqref{ndnxbound1}, \eqref{ndnxbound2}, we see that the
constant factor multiplying the volume terms on the RHS is minimized by setting
$\al=\gm^{-1/2}$, which yields the constant factor $\frac{2\gm}{(\sqrt{\gm}-1)^2}$.
\end{proof}

\vspace{-1ex}
\paragraph{Algorithm and analysis.}
The algorithm, \ndknmcalg, for \ndknmc is quite similar to \edtwonmcalg.  
The only changes are the following:
\begin{list}{$\bullet$}{\itemsep=0ex \addtolength{\leftmargin}{-2ex}}
\item In steps B1 and B3.5, we set $\T'$ to be the $s_i$-$t_i$ pairs from $\T$ that are 
at least $\gm(k-1)$-node-connected in $(S,E(S))$. 
\item In step B3.2, we apply Corollary~\ref{ndncor3} with the set $S$ to find the radius
$\rho_1\in[0,1)$.
\item In step B3.4, we add $A$ to $\Sc$, and add all nodes of $\Gm^S_i(s_i,\rho_1)$ to $Z$.  
\item Of course, in step B4, we now recursively call \ndknmcalg (with the same arguments).
\end{list}

\begin{theorem} \label{ndknmcthm}
For any $\gm>1$, algorithm \ndknmcalg returns a solution $Z$ such that 
$c(Z)\leq O\bigl(\frac{\gm}{(\sqrt{\gm}-1)^2}\ln r\ln\ln r\bigr)\cdot\OPT$ and every
$s_i$-$t_i$ pair is less than $\gm(k-1)$-node-connected in $(V\sm Z,E(V\sm Z))$.

Thus, taking $\gm=\frac{k}{k-1}$, we obtain a feasible solution of cost at most
$O\bigl((k-1)^2\ln r\ln\ln r\bigr)\cdot\OPT$.
\end{theorem} 

\begin{proof} 
Let $Z$ be the output of
$\ndknmcalg\bigl(V,\{(s_1,t_1),\ldots,(s_r,t_r)\}\bigr)$. It is clear that $Z$ is
feasible. 
Mimicking the proof of Lemma~\ref{indnlem}, and using Corollary~\ref{ndncor3} in place of 
Lemma~\ref{cor2} in the proof, one can easily show that if $d$ is the depth of the
recursion tree rooted $(U,\T)$ and $Z_U=\ndknmcalg(U,\T)$, then 
$$ 
c(Z_U)\leq
\frac{2\gm}{(\sqrt{\gm}-1)^2}\Bigl(\beta|\T|+\V^x(U)\Bigr)%
\ln\Bigl(\frac{e^dr\V^x(U)}{\OPT}\Bigr)\ln\ln\bigl(e(r+1)\bigr).
$$
Since the depth of the recursion tree is $O(\log_2 r)$, we obtain that 
$c(Z)\leq O\bigl(\frac{\gm}{(\sqrt{\gm}-1)^2}\ln r\ln\ln r\bigr)\cdot\OPT$. 
\end{proof}


\appendix

\section{\boldmath The $k$-route all-pairs cut problem} \label{append-allpairs}

\begin{theorem} \label{threeroute}
The $3$-route all-pairs cut problem is APX-hard.
\end{theorem}

\begin{proof}
We give an $L$-reduction from vertex cover on bounded-degree graphs, which is known to be
\apx-hard \cite{PY91}. 
Given a vertex-cover instance $\hat G=(\hat V, \hat E)$, where $\hat G$ has maximum degree $\alpha=O(1)$, 
we construct an instance $G=(V,E)$ of the $3$-route all-pairs cut problem. In the
following, to avoid confusion, we will refer to the elements $(\hat V, \hat E)$ of the
vertex-cover instance $\hat G$ as \emph{vertices} and \emph{edges}, and to the elements
$(V,E)$  of the constructed $3$-route all-pairs-cut instance as \emph{nodes} and
\emph{links}. 

\begin{figure}[t!]
\centering
\includegraphics[width = 0.55\textwidth]{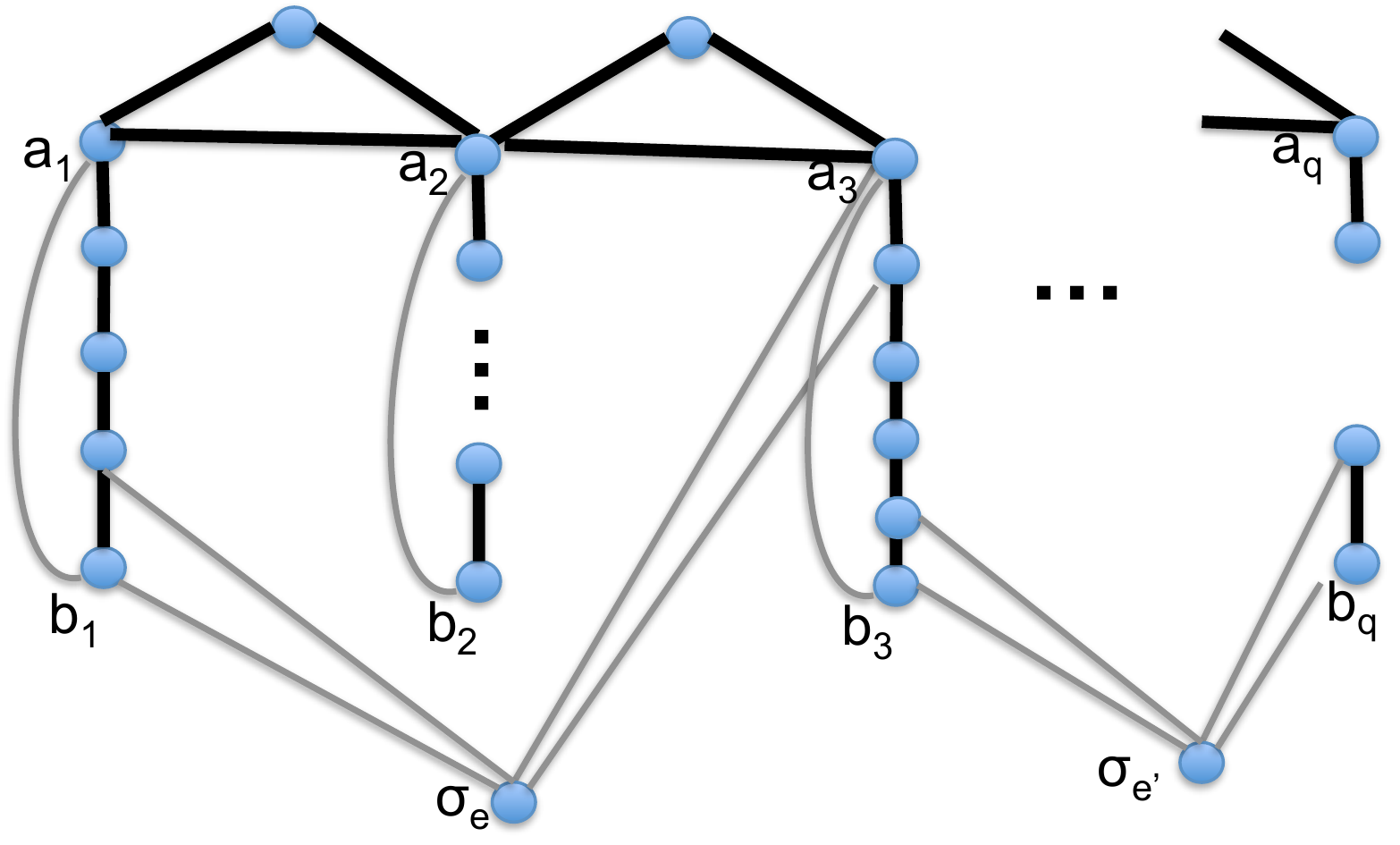} 
\caption{The instance created by our reduction from a vertex-cover instance $\hat G$ on
  $q$ vertices. Black edges have infinite cost and grey edges have unit cost.  
  Node $\cl_e$ represents the edge $(1,3)$ and node $\cl_{e'}$ represents the edge
  $(3,q)$. Vertex $1$ has degree 4 in $\hat G$ and vertex $3$ has degree 5 in $\hat G$. 
}
\label{fig:cactus}
\end{figure}

Let the vertices in $\hat V$ be numbered $1,2, \dots, |\hat V|$. 
For every vertex $v \in \hat V$, we introduce a path $P_v$ in $G$ that contains as many
links as the degree of $v$. That is, $P_v$ has one link $f_{e^v}$ for every edge 
$e \in \hat E$ incident to $v$ in $\hat G$. We give infinite cost to such links. 
Let $a_v$ be the first node of the path $P_v$ and $b_v$ be the last. We add a link
$(a_v,b_v)$ of unit cost in $G$. Note that $P_v$ and $(a_v,b_v)$ yields a cycle in $G$
for every $v \in \hat V$. 
We also connect $a_v$ to $a_{v+1}$ through a cycle formed by 3 links with infinite
cost. That is, we introduce $|\hat V| - 1$ triangles connecting all paths. 
For each edge $e =(u,v)\in \hat E$ we introduce a node $\cl_e$ and we connect $\cl_e$ 
to the endpoints of $f_{e^u}$ and to the endpoints of $f_{e^v}$, with links of unit cost. 
We let $G=(V, E)$ be the resulting graph for our $3$-route all-pairs cut instance (see Fig.\ref{fig:cactus}).

Let $p^*$ and $c^*$ be the cost of an optimal solution for the vertex-cover instance and
the cost of an optimal solution for the $3$-route all-pairs cut instance, respectively.  
We claim that: 
\begin{list}{(\roman{enumi})}{\usecounter{enumi} \topsep=0ex \itemsep=0ex
    \addtolength{\leftmargin}{-0.5ex} \addtolength{\labelwidth}{\widthof{(ii)}}}
\item If there exists a vertex cover in $\hat G$ of size $p$, then there is a solution for
  the $3$-route all-pairs cut instance of cost at most $2|\hat E| +p$. Note that this
  implies that  $c^* \leq 2|\hat E| +p^*  \leq (2\alpha + 1) p^*$.  

\item For any feasible solution for the $3$-route all-pairs cut instance of cost at most
  $2|\hat E| +p$ we can construct a cover of $\hat G$ of size at most $p$. 
\end{list}
This implies that we have an $L$-reduction, and shows that if we have a
$\beta$-approximation for $3$-route all-pairs cut, then we can obtain a vertex-cover
solution of size at most 
$\beta c^*-2|\hat E|=O(\beta)p^*$,
yielding the theorem.

In proving this, a useful observation is that a set $F$ of edges is feasible for $3$-route
all-pairs cut problem iff the remainder graph $\bG=(V,E\sm F)$ has the property that every two
simple cycles meet at most at one vertex. Such a graph is called a {\em cactus} graph.

For (i), suppose there exists a vertex cover of size $p$. For every $v$ in the cover,
we add $(a_v,b_v)$ in $F$. Furthermore, for each edge $e=(u,w) \in \hat E$, we select
one vertex between $u$ and 
$w$ that is in the cover (at least one of them is in the cover by definition), say $u$, and
we add to $F$ the links connecting $\cl_e$ to the endpoints of $f_{e^w}$ for the other vertex $w$.
It is not difficult to check that $F$ yields a feasible solution for the $3$-route all-pairs cut
instance (using the relationship to cactus graphs) of the claimed cost.

For (ii), suppose we have a feasible solution $F$ for the $3$-route all-pairs
cut instance, and consider the remainder graph obtained by removing $F$. 
Clearly, all links of infinite cost are still present.
Note that each node $\cl_e$ can have at most two links incident
to it in the remainder graph, and both these links must be incident to two nodes
of the same path $P_v$ for some $v$. If not, then we would have an
infinite-cost link of some triangle that connects the vertices $\{a_v\}_{v \in \hat V}$
that is contained into another cycle other than the triangle, contradicting feasibility of
our solution. 
We first argue that we may assume 
that each $\cl_e$ has
exactly two links incident to it in the remainder graph (and hence, also in $F$).  
Let $e=(u,w)$ be the edge in $\hat E$ corresponding to $\cl_e$. 
As argued above, $F$ contains at least one pair of links that
connect $\cl_e$ to the nodes of a path, say $P_u$.
Suppose that $F$ also contains some links connecting $\cl_e$ to $P_w$. If we
remove such links from $F$, 
we create one additional cycle, containing the link $f_{e^w}$, in the remainder
graph. Thus, the new remainder graph is not a cactus iff $f_{e^w}$ is already contained in
some cycle in $(V,E\setminus F)$.  
But there is only one possible cycle in $(V,E\setminus F)$ containing 
$f_{e^w}$, namely the cycle formed by $P_w$ and $(a_w,b_w)$. 
This observation implies that 
$F\cup\{(a_w,b_w)\}\setminus\{\text{the two links connecting $\cl_e$ to $P_w$}\}$ 
is a feasible solution to our $3$-route all-pairs cut instance. Furthermore, this solution  
has no greater cost since we are adding at most one link of unit cost, and we removing at
least one link of the same cost. 

Since the cost of our solution is at most $2|\hat E| + p$, it follows that there are at
most $p$ links in $F$ of the form $(a_v, b_v)$. We claim that these vertices $v$ form
a cover in $\hat G$. 
Suppose not. Then there is at least one edge $e \in \hat E$ that is not covered by these
vertices. 
We know that the node $\cl_e$  is connected in $\bG$ to the
endpoints of the link $f_{e^u}$ for one of the endpoints, say $u$, of $e$. The link
$f_{e^u}$ is therefore contained in the cycle formed by $P_u$ and $(a_u,b_u)$, since
$(a_u,b_u)$ is not in $F$, and is also contained in the triangle with the node $\cl_e$,
which contradicts feasibility of $F$. 
\end{proof}

\begin{corollary} \label{allk}
The $k$-route all-pairs cut problem is APX-hard for all $k\geq 3$.
\end{corollary}

\begin{proof}
The reduction is very similar to the one in the proof of Theorem~\ref{threeroute}. The
only change is that in the graph $G$ created from the given vertex-cover instance, we now have:
(a) $k-2$ parallel links between every pair of consecutive nodes of every 
 path $P_v$; and
(b) $k-2$ parallel links between $a_j, a_{j+1}$ for all $j=1,\ldots,|\hat V| -1$.

Suppose there exists a vertex cover of size $p$. As before, for each node $\cl_e$
we remove exactly one pair of links incident to $\cl_e$, and in particular we choose the
pair of links that connect  $\cl_e$ to the endpoints of the edge $f_{e^u}$ if $e=(u,v)$
and $v$ is in the cover, and the other pair otherwise. We also remove all edges of the
form $(a_v,b_v)$ for $v$ in the cover. We remove $2|\hat E|+p$ edges in total.

We claim that for every pair of nodes $z,w$ of the remainder graph, we have at most $k-1$ edge-disjoint paths. 
Every node that is not a node of a path $P_v$ has maximum degree two, and therefore this is clear.  
If $z$ and $w$ belongs to two different paths $P_v$ and $P_u$ with $u>v$, then every path connecting them must use either the link $(a_v, a_{v +1})$
(and there are at most $k-2$ such links) or the two upper links of the triangle formed with $a_v$ and $a_{v+1}$. 
Therefore, we can have at most $k-1$
edge-disjoint such paths. Finally, if  $z$ and $w$ belong to the same path $P_u$, we have $k-2$ paths given by the infinite cost links,
plus at most one additional path that uses either the edge $(a_u,b_u)$ or a sequence of pairs of links incident into the nodes $\{\cl_e\}$ for the edges $e$
that have $u$ as an endpoint in $\hat G$.
Note that, by construction, if   $(a_u,b_u)$ is still in the graph, then all the pairs of links incident into the nodes $\{\cl_e\}$ for the edges $e$
that have $u$ as an endpoint have been removed,
and therefore, once again we get at most $k-1$ edge-disjoint paths.

For the other direction, suppose we have a feasible solution $F$ for the $k$-route all-pairs
cut instance, and consider the remainder graph obtained by removing such set of links.  
Clearly, all links of infinite cost are still present.
Once again, each node $\cl_e$ corresponding to an edge $e=(u,v)$ can have at most two links incident
to it in the remainder graph, and both these links must be incident to two nodes
of the same path. If not, then we would have a pair of nodes ($a_v$ and $a_u$)
that are connected by $k$ edge-disjoint paths: 
$k-1$ given by the infinite-cost links not in the paths $P_u$  and $P_v$, and one which
uses the links in the paths $P_u$ and $P_v$  
and two links incident to $\cl_e$.
Also, as before, we may assume that each $\cl_e$ has exactly two links incident into it in
the remainder graph (and hence, in $F$), because otherwise for one endpoint $u$ of $e$ we
get that $F \cup\{(a_u,b_u)\} \setminus\{\text{the two links connecting $\cl_e$ to $P_u$}\}$
is a feasible solution to our $k$-route all-pairs cut instance of no larger cost.
So if the cost of $F$ is at most $2|\hat E| + p$, it follows that we have at most $p$
links in $F$ of the form $(a_v, b_v)$. We claim that these vertices $v$ form a 
cover in $\hat G$. 

Suppose not. Then there is at least one edge $e  \in \hat E$ that is not covered by such vertices.
We know that the node $\cl_e$  is connected in $\bG$ to the
endpoints of the link $f_{e^u}$ for one of the endpoints, say $u$, of $e$.
The endpoints of $f_{e^u}$ are therefore connected by $k-2$ parallel paths using one
single link, one path formed by the edge $(a_u,b_u)$  and edges of $P_u \setminus\{f_{e^u}\}$, 
and one path contained in the triangle with the node $\cl_e$. This contradicts feasibility
of $F$. 
\end{proof}

On the positive side, the $3$-route all-pairs cut problem admits an
$O(1)$-approximation. This follows from: 
(1) the equivalence of $3$-route all-pairs cut and the problem of removing a
min-cost set of edges so that the remainder graph is a cactus;  
(2) the results of Fiorini et al.~\cite{FioriniJP10}, who gave an $O(1)$-approximation for
the 
problem of removing a minimum-weight node set so that the remaining graph is a cactus; and 
(3) the edge-removal version easily reduces to the node-removal version by subdividing
edges, and setting the cost of the original vertices to $\infty$ and the cost of each 
vertex corresponding to an edge to be the cost of the edge.

Recently, Fomin et al.~\cite{FominLMS12} developed an $O(1)$-approximation algorithm for
the problem of removing the fewest number of nodes so that the remaining graph excludes a
minor from a given list of graphs, at least one of which should be planar. 
While $k$-route all-pairs cut 
can be stated as excluding the planar graph with $k$ parallel edges as a minor of the
remainder graph, the result of~\cite{FominLMS12} does not directly apply here.
This is because our transformation of an edge-weighted instance 
to a node weighted one introduces non-uniform node weights, whereas the algorithm
in~\cite{FominLMS12} is for uniform node weights.

\section{\boldmath Hardness of the edge-deletion $k$-route node-multicut
  problem} \label{knmc} 
Recall that in the edge-deletion $k$-route node-multicut (\edknmc) problem, we have
an undirected graph $G=(V,E)$ with nonnegative edge costs $\{c_e\}_{e\in E}$, and $r$
source-sink pairs $(s_1,t_1),\dots, (s_r,t_r)$ and an integer $k\geq 1$. 
We seek a minimum-cost set $F\sse E$ of edges such that the remainder graph 
$\bG=(V,E\sm F)$ contains at most $k-1$ node-disjoint $s_i$-$t_i$ paths for all
$i=1,\ldots,r$.  

Chuzhoy et al.~\cite{ChuzhoyMVZ12} show that \edknmc 
is hard to approximate within a factor $\Omega(k^\epsilon)$. They present a reduction from 
$3$-SAT$(5)$, which is the variant of $3$-SAT where each variable occurs in at most $5$
clauses, coupled with the parallel-repetition theorem, which is essentially a reduction 
from (the minimization version) of {\em label cover}.
However, Laekhanukit~\cite{Laekhanukit14} pointed out some subtle (but fixable) errors in
their proof and proposed a correction, but his reduction also suffers from some subtle
(again fixable) errors~\cite{Laekhanukit13}. We give a correct proof below 
via a somewhat simpler reduction than the ones in~\cite{ChuzhoyMVZ12,Laekhanukit14}.

Label cover was first introduced by Arora et al.~\cite{AroraBSS97} and has been
subsequently used as a basis for many hardness reductions (see,
e.g.,~\cite{AroraL97}). Kortsarz~\cite{Kortsarz99} presented a minimization 
version of label cover (sometimes known as \minrep) with the same hardness guarantee,
that has since found use in various network-design applications (see, e.g.,~\cite{DinitzKR12}).   

In the \minrep problem, 
we are given a bipartite graph $H=(U \cup W, F)$, two sets of labels $L_1$ (for vertices
in $U$) and $L_2$ (for vertices in $W$), and a 
constraint function for each edge $e$ defined as 
$\pi_e: L_1 \rightarrow L_2$.
A \emph{labeling} is given by specifying a set of labels $f(u) \subseteq L_1$
for every vertex $u \in U$ and a set of labels $f(w) \subseteq L_2$
for every vertex $w \in W$. 
We say that a labeling \emph{covers} an edge $e=uw \in F$ if
there exists  $a \in f(u)$ and  $b \in f(w)$
such that $\pi_e(a) = b$.
\emph{Min-Rep} asks for a labeling that covers
all the edges while minimizing $\sum_{u \in U} |f(u)| + \sum_{w \in W} |f(w)|$.

\begin{theorem}[see, e.g.,~\cite{WilliamsonS10}] \label{lchard}
There are constants $\e_0,\dt_0>0$ such that there is no polytime algorithm for
\minrep with approximation factor:
\begin{list}{--}{\topsep=0.5ex \itemsep=0ex \addtolength{\leftmargin}{-1ex}
    \addtolength{\labelwidth}{\widthof{--}}} 
\item $O\bigl(q^{\e_0}\bigr)$ unless P=NP, where $q=|L_1|+|L_2|$ is the size of the label
set; 
\item $O\bigl(\Dt^{\dt_0}\bigr)$ unless P=NP, where $\Dt$ is the maximum degree of the 
underlying graph;
\item $2^{\log^{1-\epsilon} m}$ for any constant $\e$, unless NP is contained in
  deterministic quasipolynomial time, where $m$ is the number of edges.
\end{list}
\end{theorem}

\begin{figure}[t!]
\centering
\includegraphics[width = 0.7\textwidth]{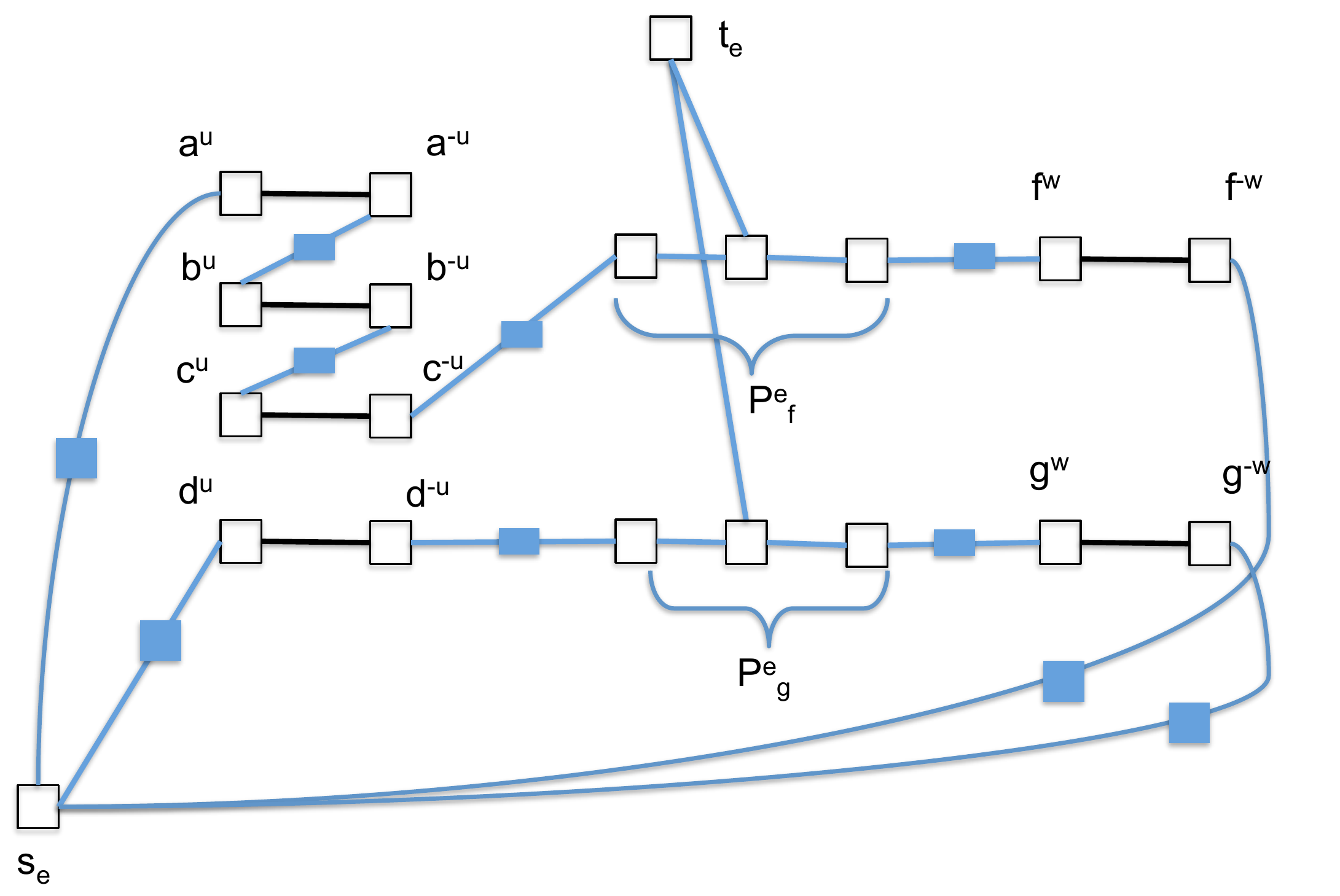} 
\caption{The gadget $(N_e,E(N_e))$ introduced for an edge $e=(u,w)$. 
  Here $L_1=\{a,b,c,d\}, L_2 = \{f,g\}, L^e_f =\{a,b,c\}, L^e_g =\{d \}$. Each black edge
  has unit cost, while all other edges have $\infty$ cost. Blue (grey) rectangles indicate
  the nodes in $V(C^e_f) \cup V(C^e_g)$.  
}
\label{fig:1}
\end{figure}

\begin{theorem} \label{edknmc-hard}
There is a polytime approximation-preserving reduction that given a \minrep-instance
\linebreak 
$(H,\pi,L_1,L_2)$ with label-size $q=|L_1|+|L_2|$ and maximum-degree $\Dt$, 
constructs an \edknmc-instance \linebreak
$\bigl(G,\{c_e\},\{s_1,t_1,\ldots,s_r,t_r\},k\bigr)$ with 
$k=O(\Dt q)$, $r=|E_H|$, and $|E_G|=O(|E_H|q)$. \\
Hence, there is no $O\bigl(k^{\e_0}\bigr)$-approximation for \edknmc, for some constant
$\e_0>0$, unless P=NP, and no $2^{\log^{1-\e}|E_G|}$-approximation for any $\e>0$ unless
NP is contained in deterministic quasipolynomial time.
\end{theorem}

\begin{proof}
We first describe the construction and then argue the approximation-preservation
property. 

\paragraph{The construction.}
For each vertex $u \in U$ and for each label in $a \in L_1$ we introduce two vertices
$a^u, \bar a^u$ in $G$ connected by an edge of unit cost. 
Intuitively, if we select this edge in an \edknmc solution for the instance we construct,
this implies that we are selecting label $a$ for $u$.
Similarly, for each $w \in W$ and for each label $b \in L_2$ we introduce two vertices
$b^w, \bar b^w$ connected by an edge of unit cost. 
The above edges will be the only ones having unit cost. All subsequent edges added to this
construction will have infinite cost.  

Consider an edge $e=(u,w)$ of $H$. For each such edge, we construct the following
gadget. We introduce two nodes $s_e,t_e$ that will form a terminal pair in our new
instance. $s_e$ and $t_e$ are connected as follows. 
For each $b \in L_2$, let $L^e_b \subseteq L_1$ be the labels of $L_1$ such that
$a \in L^e_b$ implies $\pi_e(a) = b$. 
Clearly, the sets $L^e_b$, for all $b \in L_2$, form a partition of the labels $L_1$.
For each non empty set $L^e_b$ we add a path $P^e_b$ of length 2 with the middle vertex
connected to $t_e$. We then add edges to form a cycle $C^e_b$ starting and ending at $s_e$,
containing all the edges in $L^e_b$, the edges in the path $P^e_b$ and the edge $b^w \bar b^w$.
Finally we split each of these added edges into 2 by introducing a middle vertex. We
let $V(C^e_b)$ we the the set of new vertices introduced by this operation, and let
$N_e$ be all the vertices participating in this gadget; see Fig.~\ref{fig:1}.

Let $G'=\bigl(\bigcup_{e'\in E_H}N_{e'},\bigcup_{e'\in E_H}E(N_{e'})\bigr)$ be the graph formed
by the vertices and edges of all the edge gadgets. 
For all $e$, we are going to add other edges forming paths (of length 2) between
$s_e$ and $t_e$. We add edges $(s_e,v)$, $(v,t_e)$ for all vertices 
$v\in\Gm_{G'}(N_e)$, that is, for all $v\notin N_e$ that are adjacent in some edge-gadget
to some node in $N_e$.
Note that $N_e\cap N_{e'}=\es$ 
unless $e$ and $e'$ share an endpoint in $H$, say $u$, in which case, the two gadgets
share the vertices $\{a^u, \bar a^u\}$ for all labels $a$ of $u$.
Thus, $|\Gm_{G'}(N_e)|=O(\Delta q)$.

Define $k_e:=|\{b\in L_2: L^e_b\neq\es\}|+|\Gm_{G'}(N_e)|$.
Finally, set $k:= \max_{e} k_e$.
For all edges $e$ with $k_e<k$, we add $k-k_e$ new vertices and connect these to $s_e$ and
$t_e$ (via $\infty$-cost edges). Let $G$ be the resulting graph. This concludes our
construction. 

\paragraph{Approximation preservation.}
We now argue that any feasible solution to the 
\edknmc-instance of finite cost yields a feasible solution to the \minrep-instance of no
greater cost, and vice versa. This will complete the proof. 

\smallskip

($\Rightarrow$) Let $Z$ be a solution of finite cost for our \edknmc instance. 
Consider a node $u\in V_H$ and let $L_u\in\{L_1,L_2\}$ be the label-set of $u$.
Set $f(u):= \{a \in L_u: (a^u, \bar a^u)\in Z\}$.
Clearly, the cost of the two solutions are the same. We now claim that the resulting
labeling is feasible for the label-cover instance. 
Suppose not, then there is an edge $e\in E_H$ that is not covered. By our construction,
this means that for each $b \in L_2$ with $L^e_b\neq\es$, the subgraph of the remainder
subgraph $\bG=(V_G,E_G\setminus Z)$ induced by the nodes of the cycle $C^e_b$
and $t_e$ is connected. Each such cycle $C^e_b$, yields therefore one vertex-disjoint path
in $\bG$ between $s_e$ and $t_e$. Also, all edges 
incident to $s_e$ and $t_e$ are present in $\bG$ (since they have $\infty$ cost), so all
length-2 paths in $G$ between $s_e$ and $t_e$ are still present in $\bG$. 
It follows that the vertex connectivity of $s_e$ and $t_e$ is at least $k$, a
contradiction.  

($\Leftarrow$) For the other direction, given a labeling for the label-cover instance, we
construct $Z:=\{(a^u,\bar a^u): u\in V_H,\ a \in f(u)\}$. 
Clearly, the cost of the two solutions is the same.
We claim that $Z$ is a feasible solution to our \edknmc instance.
Suppose not. Then, for some $e=(u,w)\in E_H$, the $s_e$-$t_e$ vertex connectivity in
the remainder graph $\bG=(V_G,E_G\sm Z)$ is at least $k$. Therefore we can find a set of
vertex-disjoint paths $\mathcal P$ between these vertices of size $|\mathcal P|\geq k$. 
Without loss of generality, we may assume that all the $k-|\{b\in L_2:L^e_b\neq\es\}|$
length-2 paths between $s_e$, $t_e$ are in $\mathcal P$. 
If we remove the internal nodes on these length-2 paths from $\bG$, the connected
component containing $s_e$ in the remaining portion of $\bG$ is a subgraph of the gadget
$(N_e,E(N_e))$ for edge $e$ (shown in Fig.~\ref{fig:1}). This means that this subgraph
contains at least $|\{b\in L_2:L^e_b\neq\es\}|$ $s_e$-$t_e$ vertex-disjoint paths. 
Clearly, this is only possible if, for every label $b$ such that $L^e_b\neq\es$,
either $(b^w,\bar b^w)\notin Z$ or $(a^u,\bar a^u)\notin Z$ for every $a\in L^e_b$.
But that means that the edge $e$ is not covered by our labeling, a contradiction.
\end{proof}

\end{document}